\documentclass[11pt]{amsart}
\usepackage{tikz-cd}
\usepackage{pgfplots}
\pgfplotsset{compat=1.17}
\usetikzlibrary{arrows, arrows.meta,patterns}
\tikzset{>=stealth}
\tikzcdset{arrow style=tikz, diagrams={>=stealth}}
\usepackage{amscd,amssymb,amsmath,mathtools}
\usepackage{upgreek}
\usepackage{xfp}
\usepackage{bm}
\usepackage[skip=3pt plus1pt, indent=20pt]{parskip}

\usepackage[utf8]{inputenc}
\usepackage[english]{babel}

\usepackage{quiver}
\usepackage[new]{old-arrows}

\usepackage{caption, subcaption}

\usepackage{amssymb}
\usepackage{graphicx, color}
\usepackage[all,cmtip]{xy}
\usepackage{pb-diagram,pb-xy} 
\usepackage{relsize}

\newcommand{\id}{\textbf{\textit{I}}}

\newcommand{\x}{\times}

\usepackage{thmtools}

\declaretheoremstyle[
  spaceabove=1em, spacebelow=1em,
  headfont=\bfseries,
  notefont=\bfseries, notebraces={(}{)},
  bodyfont=\itshape,
  postheadspace=0.5em,
  qed={{\scriptsize\ensuremath{\blacktriangleleft}}}
]{thmlike}

\declaretheoremstyle[
  spaceabove=1em, spacebelow=1em,
  headfont=\normalfont,
  notefont=\normalfont, notebraces={(}{)},
  bodyfont=\normalfont,
  postheadspace=1em,
  qed=$\blacksquare$
]{prooflike}

\declaretheoremstyle[
  spaceabove=1em, spacebelow=1em,
  headfont=\bfseries,
  notefont=\bfseries, notebraces={(}{)},
  bodyfont=\normalfont,
  postheadspace=0.5em,
  qed={{\scriptsize\ensuremath{\blacktriangleleft}}}
]{deflike}

\declaretheorem[style=thmlike, numberwithin=section]{theorem}
\declaretheorem[style=thmlike, sibling=theorem]{corollary}
\declaretheorem[style=deflike, sibling=theorem]{definition}
\declaretheorem[style=deflike, sibling=theorem]{lemma}
\declaretheorem[style=thmlike, sibling=theorem]{proposition}
\declaretheorem[style=deflike, sibling=theorem]{example}
\declaretheorem[style=deflike, sibling=theorem]{notation}
\declaretheorem[style=deflike, sibling=theorem]{remark}

\newcommand*\diff{\mathop{}\!\mathrm{d}}

\newcommand{\bi}[1]{{}^b\mathrm{#1}}
\newcommand{\ee}[1]{{}^E\mathrm{#1}}

\renewcommand{\vec}[1]{\bm{#1}}

\makeatletter
 \def\@textbottom{\vskip \z@ \@plus 1pt}
 \let\@texttop\relax
\makeatother

\tikzset{>=latex} 
\tikzset{declare function={
  penrose(\x,\c)  = {\fpeval{2/pi*atan( (sqrt((1+tan(\x)^2)^2+4*\c*\c*tan(\x)^2)-1-tan(\x)^2) /(2*\c*tan(\x)^2) )}};
  penroseu(\x,\t) = {\fpeval{atan(\x+\t)/pi+atan(\x-\t)/pi}};
  penrosev(\x,\t) = {\fpeval{atan(\x+\t)/pi-atan(\x-\t)/pi}};
  kruskal(\x,\c)  = {\fpeval{asin( \c*sin(2*\x) )*2/pi}};
}}
\def\tick#1#2{\draw (#1) ++ (#2:0.04) --++ (#2-180:0.08)}

\begin{document}



\title{Hamiltonian facets of classical gauge theories on $E$-manifolds}

\author{Pau Mir}
\address{Pau Mir, Laboratory of Geometry and Dynamical Systems, Departament of Mathematics, Universitat Polit\`{e}cnica de Catalunya, Avinguda del Doctor Mara\~{n}on 44-50, 08028, Barcelona 
\it{e-mail: pau.mir.garcia@upc.edu}}
\thanks{Pau Mir is supported by La Caixa Inphinit grant. The authors are supported by the project PID2019-103849GB-I00 of MCIN/ AEI /10.13039/501100011033.}

\author{Eva Miranda}
\address{Eva Miranda, Laboratory of Geometry and Dynamical Systems, Departament of Mathematics-IMTech, Universitat Polit\`{e}cnica de Catalunya, Avinguda Doctor Mara\~{n}on 44-50, 08028, Barcelona \& CRM Centre de Recerca Matem\`{a}tica, Campus de Bellaterra Edifici C, 08193 Bellaterra, Barcelona
\it{e-mail: eva.miranda@upc.edu}}
\thanks{Eva Miranda is supported by the Catalan Institution for Research and Advanced Studies via an ICREA Academia Prize 2016 and ICREA Academia Prize 2021. Eva Miranda is also supported by the Spanish State Research Agency, through the Severo Ochoa and Mar\'{\i}a de Maeztu Program for Centers and Units of Excellence in R\&D (project CEX2020-001084-M).}

\author{Pablo Nicolás}
\address{Pablo Nicolás, Laboratory of Geometry and Dynamical Systems, Departament of Mathematics, Universitat Polit\`{e}cnica de Catalunya, Avinguda del Doctor Mara\~{n}on 44-50, 08028, Barcelona
\it{e-mail:pablo.nicolas@estudiantat.upc.edu}}
\thanks{Pablo Nicolás is supported by a collaboration grant of the Spanish Ministry of Education and by a Fundació Cellex grant within the CFIS programme.}

\begin{abstract} Manifolds with boundary, with corners, $b$-manifolds and foliations model configuration spaces for particles moving under constraints and can be described as $E$-manifolds. $E$-manifolds were introduced in \cite{NestTsyganEMan} and investigated in depth in \cite{MirandaScott}. In this article we explore their physical facets by extending gauge theories to the $E$-category. Singularities in the configuration space of a classical particle can be described in several new scenarios unveiling their Hamiltonian aspects on an $E$-symplectic manifold. Following the scheme inaugurated in \cite{WeinsteinUniversal}, we show the existence of a universal model for a particle interacting with an $E$-gauge field. In addition, we generalize the description of phase spaces in Yang-Mills theory as Poisson manifolds and their minimal coupling procedure, as shown in \cite{MontgomeryThesis}, for base manifolds endowed with an $E$-structure. In particular, the reduction at coadjoint orbits and the shifting trick are extended to this framework. We show that Wong's equations, which describe the interaction of a particle with a Yang-Mills field, become Hamiltonian in the $E$-setting. We formulate the electromagnetic gauge in a Minkowski space relating it to the proper time foliation and we see that our main theorem describes the minimal coupling in physical models such as the compactified black hole.
\end{abstract}

\maketitle

\section{Introduction}

Melrose revolutionized the calculus on manifolds with boundary with a proof of the celebrated Atiyah-Patodi-Singer theorem for manifolds with boundary. In \cite{MelroseAPS} the classical proof of Atiyah-Singer theorem was twisted with a new complex of forms, and the ideas introduced there by Melrose inaugurated the theory of $b$-calculus in the nineties.

Several applications of this discovery yielded a deep understanding of deformation quantization of symplectic manifolds with boundary. Nest and Tsygan, for instance, developed in \cite{NestTsyganBMan} their deformation quantization à la Fedosov, which is controlled by a second cohomology group like in the symplectic case. This cohomology group with coefficients in the formal series of the Planck constant $\hbar$ is nothing but the celebrated $b$-cohomology of Melrose.

This modern approach to manifolds with boundary or \emph{$b$-manifolds} allows us to visualize them as manifolds with ``augmented reality'', encoded in $b$-calculus and with topology deciphered by $b$-cohomology. 

In this new scenario, the standard cotangent bundle is replaced by the $b$-cotangent bundle, which encodes information about the boundary. This divertissement can be taken to higher levels of complexity and the tangent bundle can be rapidly replaced by other Lie algebroids. Nest and Tsygan explore this idea in \cite{NestTsyganEMan} by introducing $E$-manifolds and determine the deformation quantization of such manifolds in terms of $E$-cohomology. This cohomology is also at the core of classification theorems on $E$-manifolds and in \cite{MirandaScott} it is provided a Moser theorem where deformations of the geometry are measured by the second cohomology class of the $E$-complex.

Nevertheless, those advances on the geometry of $E$-manifolds have been not been encompassed with physical interpretations of this $E$-mirror. The purpose of this article is precisely to fill in this gap in the literature and give new physical applications to the progress in the investigation on $E$-manifolds.

In this paper we develop gauge theories on $E$-manifolds by broadening the principal bundles machinery to the $E$-category. Principal connections on a principal bundle associated to a Lie group are the tool that makes it possible to model \emph{minimal coupling}, the interaction between particles and gauge fields, which is the main object of study of quantum field theory. Techniques of symplectic geometry were already used by Sternberg in \cite{SternbergMinCoup} to write the equations of motion of a ``classical particle'' in the presence of a Yang-Mills field, a special type of gauge field. These techniques, which have had a prominent role in classical gauge theories, are extended to $E$-symplectic geometry.


The extension of gauge theories to $E$-manifolds opens the study of dynamics near infinity using \emph{compactification techniques} (see \cite{MirandaOms2} and \cite{MirandaOmsPeralta}), that enables to identify escape orbits at infinity as singular periodic orbits. Other applications to dynamics and PDEs can be found in \cite{Jaquette22} and \cite{JaquetteLessardTakayasu22}. By considering gauge theories over $E$-manifolds, similar techniques can be applied to situations where the forces inducing the dynamics arise from the interaction of particles with gauge fields. In this new scenario, instead of the McGehee change of coordinates used in \cite{MirandaOms2}, we introduce a compactification of a stationary black hole following the line inaugurated by Penrose~\cite{PenroseConformal} (see example \ref{ex:penroseblackhole}), which is not conformal. We also compute the Hamiltonian version of the equations of motion of a charged particle under the interaction with an electromagnetic field in a model which incorporates the action of gravity and electromagnetic fields in a relativistic setting. Although the physical momenta are not preserved due to coupling with the electric charges, trajectories at infinity always stay at infinity.


This procedure opens a new path for applying similar techniques as in \cite{MirandaOms2} and \cite{MirandaOmsPeralta} to prove existence of \emph{escape orbits} in the relativistic scenario. However, the purpose of this article is to set up the theory as a gauge theory over $E$-manifolds and investigate its Hamiltonian facets. Hence, applications of the model to the investigation of escape orbits mimicking \cite{MirandaOmsPeralta} will be considered in a future work.

Similar compactification methods have also found application in other very different contexts such as in the study of statistical quantities for $\operatorname{SU}(2)$-Yang-Mills theories. For instance, in \cite{YangMillsT2R2} compactifications are used to model two systems in a $(3+1)$ dimensional Minkowski spacetime. The models are shown to exhibit two phase transitions related to the breaking of center symmetries along the compactified directions.

$E$-manifolds include \emph{projective foliations} (regular and singular). Our framework is appropriate for them since many examples admit a description in terms of singularities (for instance the foliation described in example \ref{example:Minkowskispace}, known as \emph{proper time foliation}). Another instance is given in \cite{YMPropTimeFoliation}, where the $\operatorname{SO}(1, 3)$ symmetry is used to characterize singularities of the $E$-structures to find solutions for the Yang-Mills equations in a $\operatorname{SO}(1, 3)$ theory. This is an example where the coupling of symmetries and singularities of the theory provides a constraint on the theory strong enough to completely determine the outcome.



Classically, the minimal coupling was worked out for any gauge group over any manifold. For the particular case of the Lie group $\mathrm{U}(1)$ and the Minkowski space, this procedure yielded Lorentz equations for a charged particle in an electromagnetic field \cite{Frankel11}. In this paper, we extend this construction to a classical particle moving on a manifold with constraints (such as boundaries, corners or foliations) and formally incarnated as an $E$-manifold. In particular, we enlarge the construction of Sternberg to this new framework.

The minimal coupling procedure makes us revisit the Marsden-Weinstein reduction in this new scenario. The first main theorem is an extension of the shifting trick to consider reduction of coadjoint orbits for Hamiltonian actions on $E$-manifolds (Theorem \ref{thm:EShiftTrick}). This result allows us to visualize the classical theorem of Weinstein (confer Theorem \ref{thm:EWeinsteinUniversal}) as a reduction at a coadjoint orbit of the dual of the Lie algebra of the associated bundle.

In classical Yang-Mills theories, the fibres of the associated bundle $P \times_G \mathfrak{g}^* \simeq \mathrm{T}^*P/G$ are the dual Lie algebra $\mathfrak{g}^*$ and the elements of these vector spaces model the charges of the theory. For instance, in a $\operatorname{U}(1)$-Yang-Mills theory we have $\mathfrak{u}^*(1) \simeq \mathbf{R}$ and we recover the notion of electric charge. For more exotic Lie groups like $\operatorname{SU}(3)$, elements of $\mathfrak{su}^*(3)$ are called \emph{color charges} and model the interactions with strong nuclear forces. On the other hand, the sections of $P \times_G \mathfrak{g}^*$ are called \emph{matter fields} because they assign a charge (an element of $\mathfrak{g}^*$) to every point in the base manifold $M$.



For a general principal bundle $P$ with group $G$, the manifold $\mathrm{T}^*P/G$ has a Poisson structure. Indeed, as observed by Montgomery \cite{MontgomeryThesis}, its symplectic leaves are the spaces investigated by Weinstein in \cite{WeinsteinUniversal}. 
In this article we materialize it in the \emph{$E$-set-up} and the associated bundle $\ee{T}^*P/G$ becomes a vector bundle over the $E$-manifold. We show that the singular dynamics are constrained by the pullback $E$-structure in the reduction of the associated bundle $\ee{T}^*P/G$. The singularities in the base manifold are, in this sense, extended naturally to singularities in the \emph{enlarged} phase space. For the particular case of $b$-manifolds, this vector bundle was already completely described for the principal bundle $G\longrightarrow G/H$ (when $H$ is a closed Lie subgroup of $G$) in \cite{BraddellKiesenhoferMiranda}. 

Finally, we visualize the $E$-cotangent bundle as a universal space to represent Hamiltonian spaces by showing the existence of a universal model for the phase space of a particle interacting with an $E$-gauge field. This picture is completed using the minimal coupling as a way to exhibit these spaces as symplectic leaves of the ambient Poisson space, as done by Montgomery \cite{MontgomeryThesis}.

\subsection*{Organization of this paper}

In section \ref{sec:motivating} we expose a series of problems arising from physics which naturally fit into the framework of $E$-manifolds. We associate an $E$-structure to each of them, which encodes the presence of singular or constrained dynamics.

In section \ref{sec:intro} we rigorously define the notion of $E$-manifold and realize different examples from the literature, as $b$-manifolds, $c$-manifolds and regular foliations, as specific instances of $E$-manifolds. These examples also cover the motivating examples arising from physics introduced in section \ref{sec:motivating}. We also review some problems where gauge theory appears in manifolds with singularities.

In section \ref{sec:prelim} we study the algebroid structure of $E$-manifolds and contextualize it within the more general setting of foliation theory. We review basic concepts in Lie algebroids theory, such as Lie algebroid cohomology, their local description and their morphisms. We also introduce prolongations or pullbacks of Lie algebroids by submersions and recall the existence of products.

In section \ref{sec:symplgeom} we present the symplectic geometry of $E$-manifolds, which is inherited from the symplectic geometry of Lie algebroids. In subsection \ref{sec:canonical} we introduce the canonical symplectic form in the $E$-cotangent bundle, which generalizes the Liouville $b$-symplectic form of the $b$-cotangent bundle (see for instance~\cite{GuilleminMirandaPires}). We also show that, unlike for general Lie algebroids, cotangent lifts of $E$-diffeomorphisms are always defined and unique. Moreover, they preserve the canonical Liouville form and, consequently, lead to Hamiltonian group actions. We conclude the section with a review of the Marsden-Weinstein reduction introduced in \cite{MarreroReductionGut} and we prove a version of the shifting trick over $E$-manifolds (Theorem \ref{thm:EShiftTrick}).

We end up this article with section \ref{sec:GaugeTheory}, where we extend the fundamental results in gauge theories over the prolongation of $E$-structures to fibre bundles. We extend the classical theorem of Weinstein~\cite{WeinsteinUniversal} as well as the isomorphism with Sternberg's phase space to gauge theories over $E$-manifolds. The results in the framework of ``cotangent bundle reduction'' as coined in \cite{MarsdenCotangent} also extend to our setting through the $E$-shifting trick. We also show that the Poisson formulation of Wong's equations done by Montgomery in ~\cite{MontgomeryThesis} and the minimal coupling procedure extend to the $E$-category (Theorem \ref{thm:MinCoupMont}). We conclude the section presenting applications of the theorems to examples in physics. 




\section{Motivating examples from Physics} \label{sec:motivating}

In the this section we describe several motivating examples where there is a natural stratification associated to the physical problem. This natural stratification often leads to a Stefan foliation. A natural framework to describe these systems is the language of $E$-manifolds. Let us start giving some motivating examples:

\begin{example}[Space of geodesics on the Lorentz plane] 
Given a general pseudo-Riemannian manifold {$(M,g)$} and consider {$\mathcal{L}$} the space of all oriented non-parametrized geodesics. The space {$\mathcal{L}$} splits as $\mathcal{L}_{\pm}$, the space of \emph{space-like geodesics} ($g(\dot{\gamma},\dot\gamma)>0$) and \emph{time-like geodesics} ($g(\dot{\gamma},\dot\gamma)<0$), and $\mathcal{L}_0$, the space of \emph{light-like geodesics} so that $g(\dot\gamma,\dot\gamma)=0$. Khesin and Tabachnikov completely described the geometry of this problem in \cite{boristabachnikov}: the sets $\mathcal{L}_\pm$ are even dimensional and {symplectic}. $\mathcal{L}_0$ can be seen as the common boundary of $\mathcal{L}_\pm$ and has an induced \emph{contact} structure.


In dimension 2 (Lorentz plane) the set of light-time geodesics is one-dimensional and {$\mathcal{L}$} inherits a global Poisson structure. As we will see below in an explicit manner this structure is a $b$-Poisson structure.

Following \cite{boristabachnikov} we can easily describe the space of 
geodesics. For the metric $\diff s^2 = \diff x \diff y$ the light-like lines are the horizontal and vertical lines. The space-like lines have positive slope and the time-like have negative slopes so each space $\mathcal{L}_\pm$ has two
components. By taking coordinates on each quadrant we can get explicit expressions. Consider for instance space-like lines having the direction in
the first quadrant. Write the unit directing vector of a line as
$(\mathrm{e}^{-u}, \mathrm{e}^u)$, with $u \in \mathbf{R}$ then the perpendicular to the line
from the origin is $(\mathrm{e}^{-u}, -\mathrm{e}^u)$. These coordinates $(u, r)$ provide local charts of $\mathcal{L}_{+}$. In the same way coordinates can be chose in $\mathcal{L}_{-}$. By equating the slope of the line $(\mathrm{e}^{-u}, \mathrm{e}^u)$ to $(1,\varepsilon)$ we obtain the relation $\varepsilon=\mathrm{e}^{2u}$. Thus, in these (singular) coordinates the symplectic structure blows up and reads $\frac{1}{\varepsilon} \diff\varepsilon \wedge \diff u$ with $u=\frac{1}{2}\log(\varepsilon)$ (see remark 2.8 in \cite{boristabachnikov}). This structure with singularities is an example of $b$-symplectic manifold.
\end{example}

\begin{example}[Compactifying the restricted three-body problem] \label{ex:McGehee}
Consider the circular planar restricted three body problem where three bodies move attracted to each other, one of them being of negligible mass as described in \cite{delshams2019global}. As it was observed in \cite{kiesenhofermirandascott}, it is possible to associate a singular structure to this problem. Consider the symplectic form on $\mathrm{T}^{\ast} \mathbb{R}^2$ in polar coordinates,
\begin{equation*}
    \omega = \diff r \wedge \diff P_r + \diff \alpha \wedge \diff P_{\alpha},
\end{equation*}
and apply to it the non-canonical McGehee change of coordinates, given by $r = \frac2{x^2}$, without altering the momentum associated to $r$.

    
Recall that, after the change to polar coordinates, the Hamiltonian associated to the restricted circular three body problem is
\[H(r,\alpha, P_r, P_{\alpha}) = \frac{P_r^2}2 + \frac{P_{\alpha}^2}{2r^2} - U(r\cos\alpha, r\sin\alpha) .\]

In the new coordinates, the Hamiltonian has the expression:

\[H(x,\alpha, P_r, P_{\alpha}) = \frac{P_r^2}2 + \frac{x^4 P_{\alpha}}8 - U\left(\frac{2\cos\alpha}{x^2}, \frac{2\sin\alpha}{x^2}\right) .\]

Furthermore, if we consider the change $r = \frac{2}{x^2}$, then $\diff r = -\frac{4}{x^3} \diff x$, and this means that

\[\omega = -\frac{4}{x^3} \diff x\wedge \diff P_r + \diff \alpha\wedge \diff P_{\alpha}.\]

Thus, the non-canonical change of coordinates transforms the symplectic form into a symplectic form that blows-up close to the line at infinity. As proved in \cite{kiesenhofermirandascott} it is a $b^3$-symplectic form.
The resulting dynamical system is nevertheless well-defined, and provides information about the original problem.

Compactifying this system with the line at infinity has the added benefit of providing a description of the dynamics within the critical set $Z=\{x=0\}$.
The dynamics within $Z$ does not have a physical meaning, but its interplay with the dynamics close to them is a way to study the behaviour of the escape orbits in this context (see also \cite{MirandaOms}). 
\end{example}

Besides configuration spaces with natural singularities, there are other physical systems which can be described in terms of appropriate $E$-manifolds.

\begin{example}[A compactification following Penrose]

The metric of Schwarzschild arises as the most general solution to Einstein's equations of motion with spherical symmetry and in the vacuum. In spherical coordinates $(t, r, \theta, \varphi)$, the metric is written as,
\begin{equation*}
    g = - \bigg( 1 - \frac{2M}{r} \bigg) \diff t^2 + \bigg( 1 - \frac{2M}{r} \bigg)^{-1} \diff r^2 + r^2 \diff \Omega^2,
\end{equation*}
where we have written $\diff \Omega^2 = \diff \varphi^2 + \sin^2 \varphi \diff \theta^2$. The coordinate $r$ is only valid in the range $2M < r < + \infty$, and this parametrization only describes the exterior of spherically symmetric objects. Our considerations are concerned with the coordinates $t$ and $r$. We consider the metric $g_\perp$, where the term containing $\diff \Omega^2$ is dropped. For convenience, we consider the auxiliary function $h(r) = 1 - \frac{2M}{r}$.

    After performing the change of coordinates $v = t + r$ and $w = t - r$, the metric reads:
\begin{equation*}
    g = \frac{1}{4} \Big( \frac{1}{h} - h \Big) \diff v^2 - \frac{1}{2} \Big( \frac{1}{h} + h \Big) \diff v \diff w + \frac{1}{4} \Big( \frac{1}{h} - h \Big) \diff v^2.
\end{equation*}
Moreover, the condition $r \geqslant 0$ is equivalent to $v \geqslant w + 4M$.

The compactification of the configuration space is achieved by defining $\alpha = \arctan v$ y $\beta = \arctan w$. The range of both coordinates is $- \pi/2 \leqslant \alpha \leqslant \pi/2$ and $-\pi/2 \leqslant \beta \leqslant \pi/2$ and equality can be attained. The compactified space is a manifold with corners. The condition $v \geqslant w + 4M$ is equivalent to $\tan \alpha \geqslant \tan \beta + 4M$. 

The region spanned by the coordinates $\alpha, \beta$ will be denoted by $N$. In these new coordinates, the metric reads
\begin{equation*}
    g_\perp = g = \frac{1}{4} \Big( \frac{1}{h} - h \Big) \sec^4 \alpha \diff \alpha^2 - \frac{1}{2} \Big( \frac{1}{h} + h \Big) \sec^2 \alpha \sec^2 \beta \diff \alpha \diff \beta + \frac{1}{4} \Big( \frac{1}{h} - h \Big) \sec^4 \beta \diff \beta^2.
\end{equation*}

%

\begin{figure}
\centering
\begin{tikzpicture}[scale=3]
\def\ta{tan(90*1.0/(4+1))}
  \def\tb{tan(90*2.0/(4+1))}
  \coordinate (O) at ( 0, 0);
  \coordinate (S) at ( 0,-1);
  \coordinate (N) at ( 0, 1);
  \coordinate (E) at ( 1, 0);
  \coordinate (X) at ({penroseu(\tb,\tb)},{penrosev(\tb,\tb)});
  \coordinate (X0) at ({penroseu(\ta,-\tb)},{penrosev(\ta,-\tb)});
  \draw[->] (-0.1,0) -- (1.2,0) node[below right=-2] {$u$};
  \draw[->] (0,-1.1) -- (0,1.2) node[left=-1] {$v$};
  \node[above=1,above left=0,blue!50!black,align=center] at (O)
    {$r=0$};
  \node[left=6,above right=-2,blue!50!black,align=center] at (1,0.04)
    {$r=+\infty$};
  \node[above=6,below right=0,align=left] at (0.04,-1)
    {$t=-\infty$};
  \node[below=6,above right=0,align=left] at (0.04,1)
    {$t=+\infty$};
  \draw[blue!80!black!60,line width=0.4] (N) -- (S);
  \foreach \i [evaluate={\c=\i/(4+1); \ct=tan(90*\c);}] in {1,...,4}{
    \draw[blue!40!red!80!black,line width=0.4,samples=30,smooth,variable=\t,domain=0.001:1]
      plot(\t,{-penrose(\t*pi/2,\ct)})
      plot(\t,{ penrose(\t*pi/2,\ct)});
    \draw[blue!80!black!60,line width=0.4,samples=30,smooth,variable=\r,domain=-1:1]
      plot({penrose(\r*pi/2,\ct)},\r);
  }
  \draw[blue!50!black] (N) -- (E) -- (S) -- cycle;
  \tick{E}{90} node[right=4,below=-1] {$+\pi/2$};
  \tick{S}{ 0} node[left=-1] {$-\pi/2$};
  \tick{N}{ 0} node[left=-1] {$+\pi/2$};
\end{tikzpicture}
\caption{The classical Penrose diagram for radius between $0$ and $+\infty$. It represents the conformal compactification of the space-time.}
\label{fig:classicalPenrose}
\end{figure}
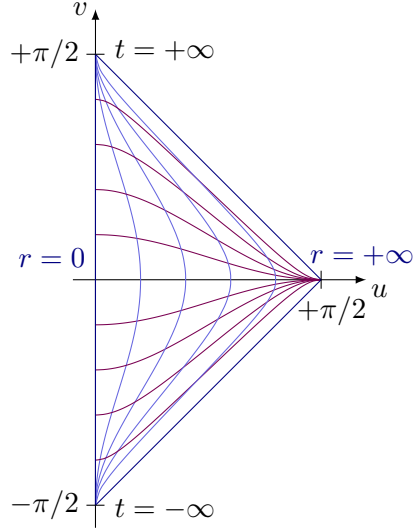

The secant function blows up quadratically at $\alpha, \beta = \pi/2$. This observation motivates us to consider  the set of vector fields which vanish quadratically at the boundary, described locally by values neighbouring $\alpha = \pi/2$. 
\end{example}

\begin{example}[The Minkowski space]
\label{example:Minkowskispace}
    Consider a Minkowski space $(M, g)$ of dimension $n + 1$ and signature $(-, +, \cdots, +)$. Consider an orthonormal basis $\{E_0, \ldots, E_n\}$ and extend it to a global chart in $M$, giving a global inertial frame of reference. Without loss of generality, assume that $E_0$ is time-like (in our convention, $g(E_0, E_0) = -1$). We can consider now the singular foliation of $\mathrm{T}M$ given by level sets of the kinetic energy, $\mathcal{F}_k = \{X \in \mathrm{T}M \mid g(X, X) = k \}$. As the norm of a geodesic is preserved, the geodesic spray is tangent to the leaves of the foliation $\mathcal{F}$. This implies that we can regard the natural phase space of the physical system $(M, g)$ as the tangent to the leaves $\mathrm{T}\mathcal{F}$, instead of $\mathrm{T}M$.
    
    The particular case of $M=\mathbf{R}^4$ with a Minkowski metric $g$ is the framework of special relativity. Given an orthonormal basis $\{E_0, E_1, E_2, E_3\}$ in which the matrix of $g$ is
    \begin{equation*}
        \begin{pmatrix} -1 & 0 & 0 & 0 \\ 0 & 1 & 0 & 0 \\ 0 & 0 & 1 & 0 \\ 0 & 0 & 0 & 1 \end{pmatrix},
    \end{equation*}
    one can construct a global coordinate system $(t, x_1, x_2, x_3)$ of $\mathbf{R}^4$ satisfying
    \begin{equation*}
        \frac{\partial}{\partial t}=E_0, \quad \frac{\partial}{\partial x_i}=E_i,\ i=1, 2, 3.
    \end{equation*}
    In these \emph{space-time coordinates}, the metric $g$ can be expressed as the $2$-form
    \begin{equation*}
        g=-(\diff t)^2 + (\diff x_1)^2 + (\diff x_2)^2 + (\diff x_3)^2.
    \end{equation*}
    The geodesics of the Minkowski space are straight lines with respect to the space-time coordinates.
    
    At any point $p\in\mathbf{R}^4$, a vector $X\in \mathrm{T}_p\mathbf{R}^4$ can be:
    \begin{enumerate}
    \item \emph{time-like}, if $g(X,X)<0$,
    \item \emph{null}, if $g(X,X)=0$,
    \item \emph{space-like}, if $g(X,X)>0$.
    \end{enumerate}
    The foliation $\mathcal{H}$ of $\mathrm{T}\mathbf{R}^4\cong \mathbf{R}^8$ given by the level sets $\mathcal{H}_k = \{X \in \mathrm{T}\mathbf{R}^4 \mid g(X, X) = k \}$ is singular. In detail, the leaves $\mathcal{H}_k$ are divided into three different types of hyper-surfaces:
    \begin{enumerate}
        \item \emph{time-like hyper-surfaces}, the two-sheet hyperboloids corresponding to $\mathcal{H}_k$ for any $k>0$.
        \item \emph{null hyper-surface}, the (singular) double cone $\mathcal{H}_0$.
        \item \emph{space-like hyper-surfaces}, the one-sheet hyperboloids corresponding to $\mathcal{H}_k$ for any $k<0$.
    \end{enumerate}
    The physical implications of the postulates of special relativity (no absolute time and no travel faster than the speed of light) are encoded in the foliation $\mathcal{H}$. Any curve $\alpha: I \to \mathbf{R}^4$ is called \emph{time-like} if $\Dot{\alpha}(t)$ is a time-like vector at $\alpha(t)$ for all $t$. Then, the world-line of any observer is a time-like curve (see Figure \ref{fig:lightcone} and the motion of an inertial (non-accelerating) observer follows a time-like geodesic. Light, on the other hand, moves on null geodesics, i.e., on the cone.
    \begin{figure}[ht!]
    \centering
    \begin{tikzpicture}[scale=0.8]
            \coordinate (A) at (0,0);
            \draw[thick] (A) +(135:5) -- +(-45:5);
            \draw[thick] (A) +(45:5) -- +(-135:5);
            \draw[dashed,->] (A) +(180:5) -- +(0:5);
            \draw[dashed,->] (A) +(270:5) -- +(90:5);
            \fill (A) circle (0.1);
            \draw[thick] let \p1=(A) in (\x1,{\y1+3.57cm}) ellipse (3.53 and 0.5);
            \draw[thick] let \p1=(A) in (\x1,{\y1-3.57cm}) ellipse (3.53 and 0.5);
            \draw[thick,variable=\t,domain=-3.7:3.7,samples=500,blue,->]
            plot ({-0.5*\t-\t*\t*0.1+\t*\t*\t*0.05},{\t});
            \node at (0,5.3) {Time};
            \node at (5.6,0) {Space};
            \node at (2.5,2) {$\mathcal{H}_0$};
            \node[blue] at (-1.5,2.5) {$\alpha(t)$};
            \node at (1.1,0.3) {Present};
            \node at (0.9,3.6) {Future};
            \node at (0.9,-3.6) {Past};
    \end{tikzpicture}
    \caption{The postulates of special relativity induce a foliation $\mathcal{H}$ in the space-time coordinates of the Minkowski space $(\mathbf{R}^4,g)$. The natural phase space of the physical system is $\mathrm{T}\mathcal{H}$, an $E$-manifold. The blue trajectory $\alpha(t)$ in the interior of the \emph{light cone} represents the world-line of an observer.}
    \label{fig:lightcone}
    \end{figure}
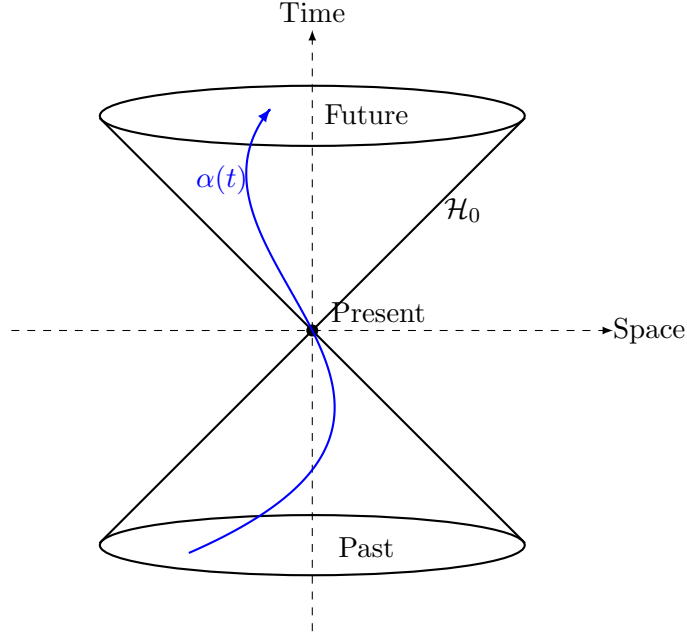
\end{example}

\begin{example}
    Spin Calogero-Moser systems are formulated on the cotangent bundle of a Lie group $G$, $\mathrm{T}^* G$. The natural phase space of a spin Calogero-Moser system is a symplectic leaf of the Poisson space $\mathrm{T}^* G/G$. The cotangent bundle can be trivialized as $\mathrm{T}^* G \simeq \mathfrak{g}^*$ using right-invariant vector fields. The Poisson space $\mathrm{T}^* G/G$ can be identified with $\mathfrak{g}^* \times_G G \simeq \mathfrak{g}^*$, and the natural projection $\pi\colon \mathfrak{g}^* \times G \longrightarrow \mathfrak{g}^*$ is a Poisson map, where the Poisson structure in $\Pi$ is the linear Poisson structure. Assume now that $\mathfrak{g}$ is semisimple, so that we have an identification $\mathfrak{g} \simeq \mathfrak{g}^*$ by means of the Killing form $\kappa$. Take now the adjoint representation of $\mathfrak{g}$ in $\mathfrak{g}^*$. The Hamiltonian function for the spin Calogero-Moser system is given by $H= \operatorname{tr}(x^2)$, where $x \in \mathfrak{g}^*$.
    
    We will now give an explicit computation of the Hamiltonian structure of spin Calogero-Moser systems for the Lie group $G = \mathfrak{su}^*(n)$. Under the identification with the Killing form, we can take elements $A \in \mathfrak{su}(n)$ as traceless and Hermitian matrices. We may trivialize $\mathrm{T}^*\mathfrak{su}^*(n) \simeq \mathfrak{su}^*(n) \times \mathfrak{su}^*(n)$, for which the moment map becomes $\mu(A, X) = [A, X]$. In parametrizing the space $\mathfrak{su}^*(n) \times \mathfrak{su}^*(n)$, we may assume that $A$ is always diagonal by conjugation with a matrix in $\mathrm{SU}(n)$. Assuming that $A = \operatorname{diag}(a_i)$ and $X = (x_{ij})$, fixing an image of the moment map $[A, X] = \mu = (\mu_{ij})$ amounts to the condition $x_{ij}(a_i - a_j) = \mu_{ij}$. Therefore, the Hamiltonian function is
    \begin{equation*}
        \operatorname{tr}(X^2) = \sum_{i,j}^n x_{ij} x_{ji} = \sum_{i = 1}^n x_{ii}^2 + \sum_{\substack{i,j = 1 \\i \neq j}}^n x_{ij}x_{ji} = \sum_{i = 1}^n x_{ii}^2 + \sum_{i,j = 1}^n \frac{\mu_{ij} \mu_{ji}}{(a_i - a_j)^2},
    \end{equation*}
    where, for convenience, we have taken $\mu_{ii} = 0$. This expression can be further simplified as $\mu_{ij} \mu_{ji} = \mu_{ij}^2 = \mu_{ji}^2$ from $\mu$ being Hermitian.
    The first term accounts for the kinetic energy of a system with $n$ different particles in a straight line, while the second term is an interaction potential dependent on the inverse square of the distance. Notice that this expression cannot be derived if $a_i = a_j$ for some $1 \leqslant i, j \leqslant n$; this remark is related to the Sjamaar-Lerman \emph{symplectic stratification}.
    
    
    For a survey on spin Calogero-Moser systems, consult \cite{ReshSCM}.
\end{example}

The previous examples show physical problems in which $E$-manifolds can be identified. The last couple of examples are specific realizations of more general phenomena.
\begin{itemize}
    \item The proper-time foliation introduced in example \ref{example:Minkowskispace} agrees with the Hamiltonian level sets $H^{-1}(k)$, for $k \in \mathbf{R}$. This is due to the fact that the Hamiltonian of the geodesic flow in a Riemannian manifold $(M, g)$ is the kinetic energy $K$. In general, the level set $H^{-1}(k)$ might not be a smooth manifold; however, we can define a regular foliation on a subset of $M$ by considering only the pre-images of regular values of $H$. This regular foliation has a naturally associated $E$-structure as we will see in example \ref{ex:RegFol}.
    \item The stratification of spin Calogero-Moser systems is a special case of the stratification of any Hamiltonian $G$-space (see, for instance, \cite{sjamaarlerman}). Assume that $(M, \omega)$ is a symplectic manifold and that $\rho\colon G \times M \longrightarrow M$ is a proper Hamiltonian action with moment map $\mu$. Sjamaar and Lerman proved in \cite{sjamaarlerman} that the reduced space $M_0$ is a stratified manifold with strata
    \begin{equation}
        M_0 = \bigsqcup_{H < G} (M_{(H)} \cap \mu^{-1}(0))/G.
    \end{equation}
    Here, the manifolds $M_{(H)}$ are the orbit types of the group action $\rho$, that is, the set of all points $p \in M$ with stabilizer $G_p$ conjugated to $H$.
\end{itemize}

\section{Preliminaries} \label{sec:intro}

We start recalling the basic definitions of $E$-manifolds. We refer the reader to \cite{NestTsyganEMan} and \cite{MirandaScott} for details.

\begin{definition} \label{def:EManDef}
    An \emph{$E$-manifold} is a pair $(M, E)$, where $M$ is a smooth manifold and $E \subseteq \operatorname{Vec}(M)$ is an involutive and locally finitely generated free $\mathcal{C}^\infty(M)$-submodule. We will call any such submodule an \emph{$E$-structure} on $M$.
\end{definition}

\begin{notation}
    In some cases, we will deal with several $E$-manifolds at the same time. To avoid any possible confusion we will commonly denote the $E$-structure of a manifold $M$ by $E_M$.
\end{notation}


\begin{example}[Manifolds with boundary, $b$-manifolds and $b^m$-manifolds]
\label{ex:bMan}
    One of the first generalizations of smooth manifolds are manifolds with boundary. Their differential calculus, called $b$-calculus, was developed by Melrose~\cite{MelroseAPS} with the aim of generalizing the Atiyah-Patodi-Singer theorem to manifolds with boundary. The differential structure of $b$-manifolds allows for a rich generalization of symplectic geometry to manifolds with boundary (see \cite{NestTsyganBMan}, for example). This realm is closer to the setting of classical symplectic geometry than to Poisson geometry; Guillemin, Miranda, and Pires proved that well-known results, like Moser's path method, also hold for manifolds with boundary~\cite{GuilleminMirandaPires}.
    
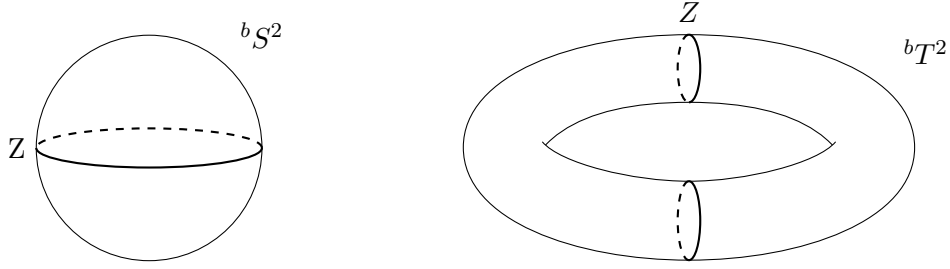
\begin{figure}[ht!]
\centering
\begin{tikzpicture}[scale=1.5]
    \def\phi{10};
    \draw (0, 0) circle (1) node[xshift=1.5cm,yshift=1.5cm] {$^bS^2$};
    \tikzset{plane/.estyle={cm={1, 0, 0, cos(90 + \phi), (0, 0)}}}
    \draw [thick,plane] (1,0) arc (0:180:1);
    \draw [thick,plane,dashed] (1,0) arc (360:180:1) node[left] {Z};
\end{tikzpicture}
\hspace{2cm}
\begin{tikzpicture}[scale=1.5]
    \draw[thick](1.5,.35) arc[start angle=90,end angle=-90,y radius=.35,x radius=.1];
    \draw[thick,dashed](1.5,.35) arc[start angle=90,end angle=-90,y radius=.35,x radius=-.1];
    \draw[thick](1.5,1.65) arc[start angle=90,end angle=-90,y radius=.3,x radius=.1] node[xshift=0cm,yshift=1.2cm] {$Z$};
    \draw[thick,dashed](1.5,1.65) arc[start angle=90,end angle=-90,y radius=.3,x radius=-.1];
    \draw (1.5,-.35) .. controls +(180:1) and +(-90:.65) .. ++(-2,1) .. controls +(90:.65) and +(180:1) .. ++(2,1) .. controls +(0:1) and +(90:.65) .. ++(2,-1) .. controls +(-90:.65) and +(0:1) .. ++(-2,-1);
    \draw (1.5,.35) .. controls +(180:.5) and +(-50:.25) .. ++(-1.3,.35) coordinate (b);
    \draw (1.5,.35) .. controls +(0:.5) and +(230:.25) .. ++(1.3,.35) coordinate (c) node[xshift=1.2cm,yshift=1.2cm] {$^bT^2$};
    \clip (1.5,.35) .. controls +(180:.5) and +(-50:.25) .. ++(-1.3,.35) -- ++(0,2) -| (c) .. controls +(230:.25) and +(0:.5) .. ++(-1.3,-.35);
    \draw (1.5,-.35) ++(0,-.6) .. controls +(180:1) and +(-90:.65) .. ++(-1.5,1) .. controls +(90:.65) and +(180:1) .. ++(1.5,1) .. controls +(0:1) and +(90:.65) .. ++(1.5,-1) .. controls +(-90:.65) and +(0:1) .. ++(-1.5,-1);
\end{tikzpicture}

\caption{The $b$-symplectic sphere $(^bS^2,Z=\{h=0\},\omega=\frac{dh}{h}\wedge d\theta)$ and the $b$-symplectic torus $(^bT^2,Z=\{\theta_1=0,\pi\},\omega=\frac{d\theta_1}{\sin(\theta_1)}\wedge d\theta_2)$ are examples of $2$-dimensional $b$-symplectic manifolds. They were studied by Olga Radko in \cite{Radko02} as compact oriented surfaces admitting a topologically stable Poisson structure.}
\label{fig:b-sphere}
\end{figure}

    The differential geometry of manifolds with boundary can be recovered from that of another similar structures, called $b$-manifolds, by a process of gluing~\cite{EvaSurgery}. A \emph{$b$-manifold} is a pair $(M, Z)$, where $M$ is a smooth manifold and $Z \subset M$ is an embedded hypersurface. Any $b$-manifold has a natural $E$-structure associated; the submodule of tangent vector fields to the submanifold $Z$. Such vectors are called \emph{$b$-vector fields}, and the set of all $b$-vector fields is denoted by $\bi{\operatorname{Vec}}(M)$. To see that $\bi{\operatorname{Vec}}(M)$ is an $E$-structure in $M$ take a point $p \in M$ and a coordinate chart $(U, \varphi)$ with coordinates $q_1, \ldots, q_n$ adapted to $Z$, that is, fulfilling $\varphi(U \cap Z) = \{q_1 = 0\}$. Under these assumptions, the local sections
    \begin{equation*}
        q_1 \frac{\partial}{\partial q_1}, \frac{\partial}{\partial q_2}, \ldots, \frac{\partial}{\partial q_n}
    \end{equation*}
    are generators of the module $\bi{\operatorname{Vec}}(U)$. This shows that $\bi{\operatorname{Vec}}(M)$ is locally free and involutive, proving that it is an $E$-structure.

    In the same way, having a smooth manifold $M$ and an embedded hypersurface $Z \subset M$ we may consider as elements of $E$ the vector fields which are tangent to $Z$ at order $m$; the set of all these fields is called the space of \emph{$b^m$-vector fields}. In a chart $(U, \varphi)$ with coordinates $\vec{q}$ adapted to $Z$, a set of local generators is
    \begin{equation}
        q_1^m \frac{\partial}{\partial q_1}, \frac{\partial}{\partial q_2}, \ldots, \frac{\partial}{\partial q_n}.
    \end{equation}
    This shows that $b^m$-vector fields give rise to different $E$-structures in $b$-manifolds.
\end{example}

\begin{example}[Manifolds with corners and $c$-manifolds] \label{ex:cMan}
    The previous example can be generalized to manifolds with intersections of higher order. We define a \emph{$c$-manifold} as a pair $(M, Z)$, where $M$ is a smooth manifold $M$ and $i\colon Z \longrightarrow M$ is an immersed hypersurface with self-transverse intersections (see Miranda and Scott~\cite{MirandaScott} for a more detailed description of the construction).
    
    By analogy with the idea that $Z$ accounts for the boundary of a manifold, self-transverse intersections can be understood as corners (hence the name $c$-manifolds). We consider the set of vector fields of $M$ which are tangent to $i(Z)$ and call them \emph{c-vector fields}. We can prove that the set of $c$-vector fields is an $E$-structure in $M$. For every point $p \in i(Z_k) \setminus i(Z_{k + 1})$ there exist coordinates $\vec{q}$ such that $i(Z) = \cup_{i \leqslant k} \{q_i = 0\}$. In these coordinates, we have the local generators
    \begin{equation*}
        q_1 \frac{\partial}{\partial q_1}, \ldots, q_k \frac{\partial}{\partial q_k}, \frac{\partial}{\partial q_{k + 1}}, \ldots, \frac{\partial}{\partial q_n}.
    \end{equation*}
    From this expression we trivially have that $c$-vector fields give rise to an $E$-structure in $M$.
\end{example}

\begin{example}[Spherical cotangent bundle compactification \cite{NestTsyganEMan}]
    Let $M$ be a smooth manifold and consider its cotangent bundle $\mathrm{T}^*M$. Over a local trivializing set $U \subset M$, we have that $\mathrm{T}^* M|_U \simeq U \times \mathbf{R}^n$. After a radial compactification of the fibre $\mathbf{R}^n$, we obtain a local description given by $U \times \mathbf{D}^n$. This process, which is a compactification of the cotangent bundle $\mathrm{T}^*M$ by the cosphere bundle $\mathrm{S}^*M$, gives as a result the closed ball bundle $\overline{\mathrm{B}}^* M$.
    
     $\overline{\mathrm{B}}^* M$  is a manifold with boundary $\partial (\overline{\mathrm{B}}^* M) = \mathrm{S}^* M$. Accordingly, we may consider the $E$-structure given by the set of vector fields tangent to the boundary (in this case, a $b$-manifold structure). One can show that the standard symplectic form on $\mathrm{T}^* M$ is naturally extended to an $E$-symplectic form over the compactification $\overline{\mathrm{B}}^* M$.
\end{example}

\begin{figure}[h!]
    \centering
    \includegraphics[scale=0.5]{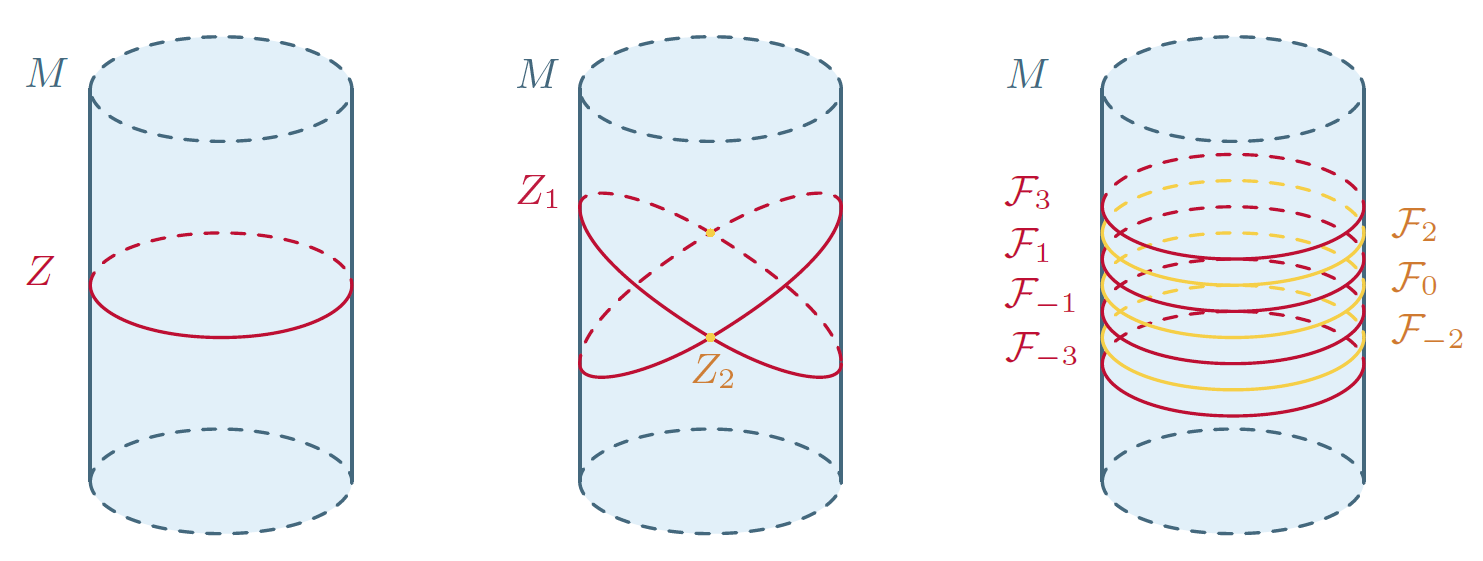}
    \label{fig:my_label}
    \caption{Different $E$-structures on the cylinder $M = \mathbf{R} \times \mathbf{S}^1$. On the left, a $b$-manifold structure taking an embedded circle $Z = \mathbf{S}^1$. On the center, a $c$-manifold structure with the strata $Z_1$, $Z_2$. Notice that the condition $Z_2 \subset \overline{Z}_1$ is satisfied. On the right, some leaves of a regular foliation.}
\end{figure}

\begin{example}
    Consider the free submodule $E \subset \operatorname{Vec}(\mathbf{R}^2)$ of vector fields generated by the fields
    \begin{equation}
        V = x \frac{\partial}{\partial x} + y \frac{\partial}{\partial y} \quad \text{and} \quad W = - y \frac{\partial}{\partial x} + x \frac{\partial}{\partial y}.
    \end{equation}
    As observed in \cite{MirandaScott} $E$ is an $E$-structure on $\mathbf{R}^2$ modelling an elliptic singularity at $(x, y) = 0$. More generally, singularities of elliptic type on an arbitrary smooth manifold $M$ are codimension 2 submanifolds with an $E$-structure on their normal bundle. These structures have been used in \cite{CavalGual} to study stable generalized complex structures.
\end{example}

\subsection{And what about singular foliations?}

Observe that the definition of $E$-manifold \ref{def:EManDef} is really similar to that of smooth singular foliation given by Androulidakis and Skandalis \cite{AndroulidakisSkandalis09}. It is natural to wonder to which extent both definitions agree and disagree. The following example shows that, in the case of regular foliations, both are the same.

\begin{example}[Regular foliations] \label{ex:RegFol}
    Consider a smooth manifold $M$ and a regular, smooth and involutive distribution $\mathcal{D}$ of rank $k$. By Frobenius' theorem, there exists a foliation $\mathcal{F}$ such that any element of $\mathcal{D}$ is tangent to a leaf of $\mathcal{F}$. The distribution $\mathcal{D}$ defines an $E$-structure by the involutivity condition. A choice of local coordinates $\vec{q}$ in an open set $U \subset M$ adapted to the foliation $\mathcal{F}$ gives a local basis
    \begin{equation*}
        \frac{\partial}{\partial q_1}, \ldots, \frac{\partial}{\partial q_p}. \qedhere
    \end{equation*}
    
    The Minkowski space (example \ref{example:Minkowskispace} in the list of motivating examples) provides a physical situation in which it is convenient to consider foliations. These foliations are often extended to consider singular foliations.
\end{example}

It is clear, therefore, that the differences between both concepts should appear in the presence of singularities. Both definitions only differ in the fact that singular foliations are locally finitely generated, while $E$-structures are locally free. The following example shows that a locally finitely generated submodule does not need to be locally free.

\begin{example}
    Let $M = \mathbf{R}^2$ and consider $\mathcal{F} \subset \operatorname{Vec}_\mathrm{c}(M)$ as the set of vector fields vanishing at $\vec{0}$. At a neighbourhood containing $\vec{0} \in U$ we have a set of generators $x \frac{\partial}{\partial x}, y \frac{\partial}{\partial x}, x \frac{\partial}{\partial y}, y \frac{\partial}{\partial y}$. On the other hand, in a neighbourhood $U_p$ with $U_p \cap \{\vec{0}\} = \varnothing$ we have generators $\frac{\partial}{\partial x}, \frac{\partial}{\partial y}$. This shows that the foliation $\mathcal{F}$, although locally finitely generated, is not free.
\end{example}

On the other hand, $E$-Manifolds are particular instances of singular foliations, which are also denominated in the literature as \emph{projective foliations}~\cite{GarmendiaThesis} or \emph{almost regular foliations}~\cite{DebordAlmostReg}. We shall give a more detailed characterization of their properties in section \ref{sec:prelim}.

\subsection{Gauge theories on $E$-manifolds}

We have seen that $E$-manifolds model degenerate or constrained physical systems. Gauge theories are ubiquitous in  theoretical physics and in particular are used to formulate electromagnetic, electroweak and strong nuclear forces.  The following examples show partial results in this direction.

\begin{example}[Geodesics in $b$-manifolds and $b^m$-manifolds]
    In \cite{MirandaOms} the authors generalize geodesic flows to $b$-manifolds and $b^m$-manifolds. The geodesic equations of motion can be framed under a gauge theory, where the principal bundle is taken to be the orthonormal bundle $\operatorname{O}(M)$. Any principal connection induces an affine connection in $\mathrm{T}M$, which gives the equations for the geodesic flow. We will generalize this construction to any $E$-manifold in section \ref{sec:GaugeTheory}, and show that geodesic flows are Hamiltonian. The Hamiltonian function, as in classical mechanics, is the kinetic energy.
\end{example}

\begin{example}[Minimal coupling in $b$-Lie groups]
    In \cite{BraddellKiesenhoferMiranda} the authors study the minimal coupling procedure of Montgomery for $b$-Lie groups. Such groups are pairs $(G, H)$, where $G$ is a Lie group and $H$ is a closed codimension one subgroup. This structure is an example of a $b$-manifold in the sense of example \ref{ex:bMan}; here, the natural projection $\pi\colon G \longrightarrow G/H$ is an $H$-principal bundle. The quotient group is induced a $b$-structure taking as singular hypersurface the class $Z_H = [1]$, which has codimension 1 given that $\dim(G/H) = 1$. Additionally, with this definition the map $\pi$ is a $b$-map. Some highlighted examples of $b$-Lie groups are the Galilean group or the Heisenberg group, where the subgroup $H$ is identified with the set of time-preserving transformations. The minimal coupling procedure in this article is a particular case of the one exhibited in theorem \ref{thm:MinCoupMont}.
\end{example}

\section{The geometry of $E$-manifolds and associated bundles} \label{sec:prelim}

After having described examples where $E$-manifolds naturally appear, we review the basic definitions and results in the theory of $E$-manifolds. The main definitions are present in \cite{MirandaScott}, and many constructions are borrowed from the literature of Lie algebroids (see, for instance part VII in \cite{CannasWeinstein99}).

\subsection{Local structure, exterior differential and cohomology}

We begin by showing that the information of an $E$-structure can be encoded in a Lie algebroid, called the $E$-tangent bundle. This characterization allows us to apply geometric arguments similar to those in classical mechanics on the standard tangent bundle $\mathrm{T}M$. 

The key ingredient of this construction is the following:

\begin{theorem}[Serre-Swan theorem \cite{swan}] \label{thm:serreswan} 
    A $\mathcal{C}^{\infty}(M)$-module $P$ is isomorphic to the module of sections of a vector bundle $E$ (denoted as $\Gamma(E)$) if and only if $P$ is finitely generated and projective.
\end{theorem} 

The main idea is that projective modules over commutative rings are like vector bundles on compact spaces. For practical reasons,  we will focus on free modules. The  condition of $E$ being locally free implies that theorem \ref{thm:serreswan} holds and elements of $E \subset \operatorname{Vec}(M)$ are in bijection with sections of a vector bundle $\ee{T}M$, called the \emph{$E$-tangent bundle}. Localizing this equivalence of sections, we can define a vector bundle map $\rho_E\colon \ee{T}M \longrightarrow \mathrm{T}M$. The anchor map $\rho_E\colon \ee{T}M \longrightarrow \mathrm{T}M$ is defined by inclusion of the set of sheaf of local sections. Moreover, as the distribution $E$ is involutive,  the Lie bracket of vector fields can be restricted to sections of $\ee{T}M$. In this way, the triple $(\ee{T}M, \rho, [\cdot, \cdot])$ becomes a Lie algebroid. 

We present now some additional properties satisfied by the anchor map $\rho$ that are useful to give pictorial examples. This result was already noticed in \cite{GarmendiaThesis} and dates back to \cite{DebordAlmostReg}.

\begin{lemma}
    The anchor map of the Lie algebroid $\ee{T}M $, $\rho_E\colon \ee{T}M \longrightarrow \mathrm{T}M$ is generically injective.
\end{lemma}

\begin{proof}
    For the sake of clarity, we shall review the construction of the anchor map $\rho_E$ pointwise. By Serre-Swan theorem \ref{thm:serreswan}, for every point $p \in M$ there exists an open set $U_p \subset M$ and a ring isomorphism $\tilde{\rho}_E\colon \Gamma_{U_p}(\ee{T}M) \longrightarrow E_{U_p}$. We define $\mathcal{I}_p = \{ f \in \mathcal{C}^\infty(M) \mid f(p) = 0 \}$, the ideal of vanishing smooth functions at $p$. The evaluation map $\operatorname{ev}_p\colon E_{U_p} \longrightarrow \mathrm{T}_p M$ is well defined in the quotient ring $E_{U_p}/\mathcal{I}_p E_{U_p}$, giving a map $\widetilde{\operatorname{ev}}_p\colon E_{U_p}/\mathcal{I}_p E_{U_p} \longrightarrow \mathrm{T}_p M$. The ring isomorphism $\tilde{\rho}_E$ induces now a ring isomorphism $\Gamma_{U_p}(\ee{T}M)/\mathcal{I}_p \Gamma_{U_p}(\ee{T}M) \simeq E_{U_p}/\mathcal{I}_p E_{U_p}$ which, by precomposition, defines a map $\alpha\colon \Gamma_{U_p}(\ee{T}M)/\mathcal{I}_p \Gamma_{U_p}(\ee{T}M) \longrightarrow \mathrm{T}_p M$. Taking the canonical identification $\Gamma_{U_p}(\ee{T}M)/\mathcal{I}_p \Gamma_{U_p}(\ee{T}M) \simeq \ee{T}_p M$ we finally have $\rho_E|_p \coloneqq \alpha\colon \ee{T}_p M \longrightarrow \mathrm{T}_p M$. Notice that, by construction, $\tilde{\rho}_E = \Gamma(\rho_E)$.

    Let $K \subset M$ be the set of points $p \in M$ for which $\ker \rho_p \neq \varnothing$. As both bundles have the same rank, this set is the set of points $p \in M$ such that $\det \rho_p \neq 0$. Assume now that $M \setminus K$ were not dense (i.e., $\operatorname{int} K \neq \varnothing$). Take an open set $U \subset K$ contained in the kernel. By smoothness, there exists a point $p \in U$ and an open neighbourhood $V_p \subset U$ such that the rank of $\rho_E$ is constant; therefore, $\ker \rho_E$ is a vector bundle over $V_p$ and, by shrinking the domain if necessary, we may assume that it is trivial and that the Serre-Swan property in theorem \ref{thm:serreswan} is satisfied. But now, a basis $F_1, \ldots, F_l$ of $\Gamma_{V_p}(\ker \rho_E)$ can be extended to a set of generators $F_1, \ldots, F_l, E_{l + 1}, \ldots, E_k$ of $\Gamma_{V_p}(\ee{T}M)$ and, as $\tilde{\rho}_E = \Gamma(\rho_E)$, we have $\tilde{\rho}_E(F_i) = 0$. This would contradict theorem \ref{thm:serreswan}. Thus, the map $\rho\colon \ee{T}M \longrightarrow \mathrm{T}M$ is generically injective.
\end{proof}



This characterization implies that the Stefan foliation associated to the anchor map of the algebroid $\ee{T}M$ is an instance of an almost regular foliation as coined by Debord~\cite{DebordAlmostReg}. As an application of theorem 1 in \cite{DebordAlmostReg}, the algebroid $\ee{T}M$ can be integrated to a Lie groupoid. It is certainly not the case that any singular foliation can be recovered by the anchor of a Lie algebroid (confer lemma 1.3 in \cite{AndroulidakisZambon}). We discuss now an explicit construction of the $E$-tangent bundle over $b$-manifolds and its anchor, together with an example over regular foliations.

\begin{example}
    Let us consider a $b$-manifold $(M, Z)$ following the notation of example \ref{ex:bMan}. In a local chart $U$ around a point $p \in Z$ compatible with $Z$ we have taken as generators of the $b$-tangent bundle the local fields
    \begin{equation*}
        q_1 \frac{\partial}{\partial q_1}, \frac{\partial}{\partial q_2}, \ldots, \frac{\partial}{\partial q_n}.
    \end{equation*}
    The identification of the elements of $\bi{\operatorname{Vec}}(M)$ with sections of $\bi{T}M$ implicitly uses $\tilde{\rho}_E$. Under this local description of $\tilde{\rho}_E$, for a point $q \in U_p$ with coordinates $\vec{q}$, the anchor $\rho_E|_q$ applied to an element $X = \sum_{i = 1}^n \alpha_i E_i(\vec{q}) \in \bi{T}_q M$ can be computed as
    \begin{equation*}
        \rho_E|_q(X) = \alpha_1 q_1 \frac{\partial}{\partial q_1}(\vec{q}) + \sum_{i = 2}^n \alpha_i \frac{\partial}{\partial q_i}(\vec{q}).
    \end{equation*}
    In particular, if $q_1 = 0$ (that is, if $q \in Z$) then $\rho_E|_q$ fails to be injective. The reason is that, although we are identifying sections of the $b$-tangent bundle with $b$-vector fields in $M$, the rank of the set of $b$-vector fields as sections of the tangent bundle, evaluated at points of the critical hypersurface is $n-1$. The rank  of the distribution generated by $b$-vector fields as sections of the $b$-tangent bundle is constant at all points and equal to $n$.
    
    In the setting of $b$-symplectic geometry, the critical hypersurface $Z$ is defined as the zero section of the map $\Pi^n$. In this case, there are two different algebroid structures. On one hand, the $E$-manifold structure $\rho_E\colon \bi{T}M \longrightarrow \mathrm{T}M$ constructed using Theorem \ref{thm:serreswan} in example \ref{ex:bMan}.  On the other hand, the Lie algebroid structure whose anchor is given by the Poisson tensor $\Pi^\sharp\colon \mathrm{T}^*M \longrightarrow \mathrm{T}M$. We can easily argue that both structures are different by looking at their characteristic foliations on $M$. For the $b$-tangent bundle $\bi{T}M$, the only leaves of its characteristic foliation are $Z$ and the connected components of $M \setminus Z$. Observe that $Z$ has to be odd-dimensional. The characteristic foliation of $\Pi^\sharp\colon \mathrm{T}^*M \longrightarrow \mathrm{T}M$ is the symplectic foliation of $(M, \Pi)$ and, as a consequence, its leaves are even-dimensional. This symplectic foliation can be restricted to the leaf $Z$, where it induces a codimension 1 foliation (for more details, refer to \cite{GuilleminMirandaPires}). This shows that the characteristic foliations do not agree and that the algebroid structures are consequently not isomorphic.
\end{example}

\begin{example}
    Regular foliations can be naturally endowed with an $E$-manifold structure following example \ref{ex:RegFol}. In a foliated chart $U$ with coordinates $\vec{q}$, the local fields
    \begin{equation*}
        \frac{\partial}{\partial q_1}, \ldots, \frac{\partial}{\partial q_p}
    \end{equation*}
    can be identified with generators $E_1, \ldots, E_p$ of the $E$-tangent bundle $\ee{T}M$ under the ring isomorphism $\tilde{\rho}_E$. Following the previous computations in the setting of a $b$-manifold, for a point $q \in U$ with coordinates $\vec{q}$ anchor map $\rho_E|_q$ can be computed in an element $X = \sum_{i = 1}^p \alpha_i E_i(\vec{q}) \in \ee{T}_q M$ as
    \begin{equation*}
        \rho_E|_q(X) = \sum_{i = 1}^p \alpha_i \frac{\partial}{\partial q_i}(\vec{q}).
    \end{equation*}
    Notice that, in this case, the anchor map is injective for every $q \in U$ and, as there exists a covering of $M$ by foliated charts, it is also injective for any $q \in M$. This example shows that the failure of injectiveness is closely related to the change of dimensions in the leaves of the characteristic foliation of the $E$-structure.
\end{example}

We present now a description of the local structure of $E$-manifolds. The following concepts are directly inherited from the literature of Lie algebroids.

\begin{definition}[Structure functions] \label{def:StrucConst}
    Consider an $E$-manifold $(M, E_M)$ and a coordinate set $U \subset M$ with coordinates $q_1, \ldots, q_n$. Let $E_1, \ldots, E_p$ be a set of generators of $\ee{\operatorname{Vec}}(U)$. The \emph{structure functions} are the smooth functions $\rho_{ij} \in \mathcal{C}^\infty(U)$ and $C_{ij}^k \in \mathcal{C}^\infty(U)$ satisfying
    \begin{equation}
        E_i = \sum_{j = 1}^n \rho_{ij} \frac{\partial}{\partial q_j}, \quad [E_i, E_j] = \sum_{k = 1}^p C_{ij}^k E_{k},
    \end{equation}
    where $i, j, k = 1, \ldots, p$.
\end{definition}

\begin{remark}
    The structure functions are well-defined by the involutivity of the $E$-structure, $[E, E] \subset E$. Moreover, they are skew-symmetric, $C_{ij}^k = -C_{ji}^k$, by the skew-symmetry of the Lie bracket.
\end{remark}


The \emph{$E$-cotangent bundle} (denoted by $\ee{T}^* M$) is defined as the dual bundle of $\ee{T}M$. The global sections of $\bigwedge^p \ee{T}^* M$ are called \emph{$E$-forms of degree $p$} and are denoted by $\ee{\Omega}^p(M)$. We can define a differential $\diff\colon \ee{\Omega}^p(M) \longrightarrow \ee{\Omega}^{p + 1}(M)$ explicitly as
\begin{align} \label{eq:Ediff}
    \diff \omega(X_0, \ldots, X_p) &= \sum_{i = 0}^p (-1)^i \mathcal{L}_{i(X_i)} \omega\Big( X_0, \ldots, \widehat{X}_i, \ldots, X_p \Big) \\
    & \qquad + \sum_{0 \leqslant i < j \leqslant p} (-1)^{i + j} \omega\Big( [X_i, X_j], X_0, \ldots, \widehat{X}_i, \ldots, \widehat{X}_j, \ldots, X_p \Big). \nonumber
\end{align}
The Lie derivative of a form $\omega \in \ee{\Omega}^p(M)$ along a section $X \in E$ is defined using Cartan's formula
\begin{equation} \label{eq:ELieDer}
    \mathcal{L}_X \omega = \diff \iota_X \omega + \iota_X \diff \omega.
\end{equation}
As $E$ is an involutive submodule of $\operatorname{Vec}(M)$, the Lie bracket is the restriction of the Lie bracket for vector fields of the tangent bundle $\mathrm{T}M$.

The differential \eqref{eq:Ediff} satisfies the cochain condition $\diff^2 = 0$, giving $\ee{\Omega}^\bullet(M)$ a cochain complex structure. The cohomology of this complex is called the \emph{$E$-cohomology} ${}^E H^\bullet(M)$. This cohomology can be easily identified in many of the provided examples. For $b$-manifolds, the corresponding $b$-cohomology has been largely studied in the literature. For regular foliations, this cohomology is yet another reincarnation of  \emph{foliated cohomology}.

\subsection{$E$-Maps and flows of $E$-fields}

Having defined $E$-manifolds, we are interested in morphisms between them. This leads us to the following definition.

\begin{definition} \label{def:EMap}
    Let $(M, E_M)$ and $(N, E_N)$ be $E$-manifolds. An \emph{$E$-manifold map} between $M$ and $N$ is a Lie algebroid map between $\ee{T}M$ and $\ee{T}N$; that is, it is a pair $(f, F)$ with $f \colon M \longrightarrow N$ and $F\colon \ee{T}M \longrightarrow \ee{T}N$ smooth maps such that the following diagrams are commutative
    \begin{equation} \label{eq:EMapDiag}
    \begin{tikzcd}[sep = large]
    	{\ee{T}M} & {\ee{T}N} && {\ee{T}M} & {\ee{T}N} \\
    	M & N && {\mathrm{T}M} & {\mathrm{T}N}
    	\arrow["f"', from=2-1, to=2-2]
    	\arrow["{\diff f}"', from=2-4, to=2-5]
    	\arrow["F", from=1-4, to=1-5]
    	\arrow["{\rho_M}"', from=1-4, to=2-4]
    	\arrow["{\rho_N}", from=1-5, to=2-5]
    	\arrow["F", from=1-1, to=1-2]
    	\arrow["{\tau_{\ee{T}M}}"', from=1-1, to=2-1]
    	\arrow["{\tau_{\ee{T}N}}", from=1-2, to=2-2]
    \end{tikzcd}
    \end{equation}
    and satisfying $\diff F^* = F^* \diff$.
\end{definition}

\begin{remark}
    As a consequence of the definition of a Lie algebroid map, we have the condition $\diff_p f (E_M(p)) \subset E_N(f(p))$. This is a result of commutative diagram \ref{eq:EMapDiag}. 
\end{remark}

So far, we have visualized the $E$-tangent bundle as a replacement of the standard tangent bundle $\mathrm{T}M$, yielding constrained or singular dynamics. The equations of motion in classical mechanics are specified in the form of vector fields, and their flow gives the dynamical evolution of the system. It is a surprisingly non-trivial fact that flows of $E$-fields preserve the submodule of vector fields by pushforward (confer \cite[proposition 1.6]{AndroulidakisSkandalis09}).

\begin{proposition}
    If $(M, E)$ is an $E$-manifold, the flow $\varphi_t$ of any section $X \in E$ is an $E$-map, that is, $\diff \varphi_t (E) \subset E$.
\end{proposition}

This result shows that the mechanics induced by an $E$-field restrict to the $E$-tangent bundle. A more general result concerning the lift of a vector field to an arbitrary Lie algebroid can be found in \cite{RuiLoja}. Consequently, sections of Lie algebroids can be integrated to produce one-parameter groups of transformations. This construction is vacuous over $E$-manifolds given that, by definition, sections of the $E$-tangent bundle are in correspondence with $E$-vector fields on $M$.


\subsection{Pullbacks, local structure and products}

Given an $E$-manifold, which we assume to be the natural phase space of a degenerate physical system, we expect to describe Hamiltonian mechanics in the dual bundle $\ee{T}^* M$. In order to reflect the degeneracy of $E$ in the dynamics on $\ee{T}^*M$, it is necessary to introduce the notion of \emph{prolongation}. This procedure induces an $E$-manifold structure in the $E$-cotangent bundle $\ee{T}^*M$. More importantly, this idea will be used in section \ref{sec:GaugeTheory} for a similar purpose on a principal $G$-bundle. 

\begin{definition} \label{def:ProlFibreBun}
    Let $\ee{T}M$ be an $E$-manifold and consider a fibre bundle $\tau\colon B \longrightarrow M$. The \emph{pullback} of $\ee{T}M$ to $B$ is the $E$-manifold defined as
    \begin{equation}
        \ee{T}B = \{ (u, v) \in \mathrm{T}B \times \ee{T}M \mid \diff \tau (u) = \rho_{E} (v) \}. \qedhere
    \end{equation}
\end{definition}

Notice that the pullback $E$-manifold is well-defined because $\tau\colon B \longrightarrow M$ is a surjective map, and therefore it is transversal to the anchor $\rho_E$. The following proposition shows that $\ee{T}B$ is indeed an $E$-manifold by explicitly describing the corresponding $E$-structure.

\begin{proposition}
    Let $(M, E_M)$ be an $E$-manifold, let $\tau\colon B \longrightarrow M$ be a fibre bundle and consider the pullback $E$-manifold $\ee{T}B$. Then, $\ee{T}B$ is isomorphic to the $E$-tangent bundle generated by the $E$-structure
    \begin{equation}
        E_B = \langle (\diff \tau)^{-1}(E_M) \rangle. \qedhere
    \end{equation}
\end{proposition}

\begin{proof}
    The idea of the proof is to give a local description of the set of sections of $\ee{T}B$. To do this, first we choose a connection $\mathrm{T}B \simeq \mathrm{V}B \oplus \mathrm{H}B$; under this trivialization the map $\diff \tau$ becomes the projection onto the second component. Therefore, $\mathrm{T}B \times \ee{T}M \simeq (\mathrm{V}B \oplus \mathrm{H}B) \times \ee{T}M$ and the condition $\diff \tau(u) = \rho_E(v)$ completely determines the horizontal component of $u$. As a consequence, $\ee{T}B \simeq \mathrm{V}B \times \ee{T}M$.
    
    Let us consider a point $m \in M$ and a trivializing coordinate chart $U \subset M$ centered around $m$ with coordinates $\vec{q}$. Considering the isomorphism $\tau^{-1}(U) \simeq U \times Q$, a local chart $V \subset Q$ with coordinates $\vec{r}$ gives a local chart $U \times V$ around $B$. Notice that the tangent space to $V$ becomes a local chart of $\mathrm{V}B$. If $E_1, \ldots, E_p$ is a local set of generators of $E_M(U)$ and $V_1, \ldots, V_q$ is a local basis of sections of $\mathrm{V}B$ around $V$, then $E_1, \ldots, E_p, V_1, \ldots, V_q$ is a local set of generators of $\mathrm{V}B \times \ee{T}M \simeq \ee{T}B$. One can observe now that, with our definitions,
    \begin{equation*}
        \langle (\diff \tau)^{-1} (E_M) \rangle = \langle E_1, \ldots, E_p, V_1, \ldots, V_q \rangle.
    \end{equation*}
    By Serre-Swan theorem \ref{thm:serreswan}, we conclude that $\ee{T}B$ is the $E$-tangent bundle generated by $E_B$.
\end{proof}

In the proof of the previous result we have introduced a set of local coordinates to describe the $E$-structure in $B$ which have been, in a sense, induced from a choice of coordinates in $M$ and in a fibre of $B$. We will give a complete description of this procedure in the following proposition, which can be regarded as a natural generalization of the introduction of natural coordinates in the cotangent bundle.

\begin{proposition} \label{prop:ENatCoord}
    Let $(M, E_M)$ be an $E$-manifold and let $\tau\colon B \longrightarrow M$ be a principal bundle with typical fibre $Q$ endowed with the pullback structure $E_B$. Let $U \subset M$ be a local chart satisfying the trivialization property $\tau^{-1}(U) \simeq U \times Q$ with coordinates $\vec{q}$, let $E_1, \ldots, E_p$ be local generators of $\ee{T}U$ with structure functions $\rho_{ij}$ and $C_{ij}^k$ and consider a local chart $V \subset Q$ with coordinates $\vec{r}$. If $\vec{p}$ are the coordinates induced by the sections $E_1, \ldots, E_p$ and $\vec{s}$ are the natural coordinates from $\vec{r}$, then the open set $U \times V$ with coordinates $\vec{q}, \vec{p}, \vec{r}, \vec{s}$ is a chart of $\ee{T}(U \times V)$ and $E_1, \ldots, E_p, V_1, \ldots, V_q$ are the local sections associated to $\vec{p}$ and $\vec{s}$, respectively. Moreover, we have
    \begin{equation} \label{eq:EProlFunDiff}
        \diff f = \sum_{i = 1}^k \sum_{j = 1}^n \rho_{ij} \frac{\partial f}{\partial q_j} E_i^* + \sum_{i = 1}^q \frac{\partial f}{\partial r_i} V_i^*
    \end{equation}
    Additionally, the differential of the dual sections is given by
    \begin{equation} \label{eq:EProlSecDiff}
        \diff E_i^* = - \frac{1}{2} \sum_{j,l = 1}^k C_{jl}^i E_j^* \wedge E_k^*, \quad \diff V_i^* = 0. \qedhere
    \end{equation}
\end{proposition}

\begin{proof}
    The product $U \times V$ is a local chart of $B$ with coordinates $\vec{q}, \vec{r}$. In local coordinates, the bundle projection $\tau\colon U \times V \longrightarrow U$ is expressed as $\tau(\vec{q}, \vec{r}) = \vec{q}$. From the definition of pullback structure, $E_B(U \times V) = E_M(U) \times \operatorname{Vec}(V)$ and, as a consequence, we have locally $\ee{T}(U \times V) \simeq \ee{T}U \times \mathrm{T}V$.
    
    Take now local generators $\langle E_1, \ldots, E_p \rangle$ of $\ee{T}U$, which induce coordinates $\vec{q}$, $\vec{p}$ in $\ee{T}U$. Consider now the local coordinates $\vec{r}$, $\vec{s}$ in $\mathrm{T}V$ and let $V_1, \ldots, V_q$ be a basis of sections associated to these coordinates. As a consequence of the preceding discussion, the coordinates $\vec{q}, \vec{r}, \vec{p}, \vec{s}$ are associated to a local chart of $\ee{T}(U \times V)$ with local sections $E_1, \ldots, E_p, V_1, \ldots, V_q$. From this definition we have
    \begin{equation*}
        E_i = \sum_{j = 1}^n \rho_{ij} \frac{\partial}{\partial q_j}, \quad [E_i, E_j] = \sum_{k = 1}^p C_{ij}^k E_k, \quad \text{and} \quad V_i = \frac{\partial}{\partial r_i}, \quad [V_i, V_j] = 0,
    \end{equation*}
    and the cross Lie brackets $[E_i, V_j]$ vanish. We compute the contraction of the differential $\diff f$ with the local generators, obtaining
    \begin{equation*}
        \langle \diff f, E_i \rangle = \mathcal{L}_{\sum_{j = 1}^n \rho_{ij} \partial_{q_j}} f = \sum_{j = 1}^{n} \rho_{ij} \frac{\partial f}{\partial q_j}, \quad \langle \diff f, V_i \rangle = \mathcal{L}_{\partial_{r_i}} f = \frac{\partial f}{\partial r_i}.
    \end{equation*}
    Equation \eqref{eq:EProlFunDiff} is a direct consequence of these computations, as
    \begin{equation*}
        \diff f = \sum_{i = 1}^p \sum_{j = 1}^n \rho_{ij} \frac{\partial f}{\partial q_j} E_i^* + \sum_{i = 1}^{q} \frac{\partial f}{\partial r_i} V_i^*.
    \end{equation*}
    
    To obtain equations \eqref{eq:EProlSecDiff} we have to compute the contraction of $\diff E_i^*$ following equation \eqref{eq:Ediff}. A direct computation shows
    \begin{align*}
        \diff E_i^* (E_j, E_k) &= \mathcal{L}_{E_j} \langle E_i^*, E_k \rangle - \mathcal{L}_{E_j} \langle E_i^*, E_k \rangle - \langle E_i^*, [E_j, E_k] \rangle \\
        &= - C_{jk}^i, \\
        \diff E_i^* (E_j, V_k) &= \mathcal{L}_{E_j} \langle E_i^*, V_k \rangle - \mathcal{L}_{V_k} \langle E_i^*, E_j \rangle - \langle E_i^*, [E_j, V_k] \rangle \\
        &= 0, \\
        \diff E_i^* (V_j, V_k) &= \mathcal{L}_{V_j} \langle E_i^*, V_k \rangle - \mathcal{L}_{V_k} \langle E_i^*, V_j \rangle - \langle E_i^*, [V_j, V_k] \rangle \\
        &= 0.
    \end{align*}
    As a consequence, we have
    \begin{equation}
        \diff E_i^* = - \frac{1}{2} \sum_{j,k = 1}^p C_{jk}^i E_j^* \wedge E_k^*.
    \end{equation}
    A similar computation shows that $\diff V_i^* = 0$.
\end{proof}

We will eventually have to consider the product of two $E$-manifolds. The following proposition states that the obvious product structure fulfills the desired universal property.

\begin{proposition} \label{prop:UnivPropEProd}
    Let $(M, E_M)$ and $(N, E_N)$ be two $E$-manifolds. The product $E_M \times E_N$ is an $E$-structure in $M \times N$, and the $E$-manifold $(M \times N, E_M \times E_N)$ satisfies the product universal property: for every $E$-manifold $(L, E_L)$ and every pair of $E$-maps $f\colon L \longrightarrow M$ and $g\colon L \longrightarrow N$ there exists a unique $E$-map $h\colon L \longrightarrow M \times N$ making diagram \eqref{eq:ProdDiag} commute.
    \begin{equation} \label{eq:ProdDiag}
        \begin{tikzcd}[sep=large]
        	&& M \\
        	L & {M \times N} \\
        	&& N
        	\arrow["h", from=2-1, to=2-2]
        	\arrow["{p_1}", from=2-2, to=1-3]
        	\arrow["{p_2}"', from=2-2, to=3-3]
        	\arrow["f", curve={height=-18pt}, from=2-1, to=1-3]
        	\arrow["g"', curve={height=18pt}, from=2-1, to=3-3]
        \end{tikzcd} \qedhere
    \end{equation}
\end{proposition}

\section{Symplectic geometry on $E$-manifolds} \label{sec:symplgeom}

Having introduced the essential tools in the differential geometry of $E$-manifolds, we are in position to describe their symplectic geometry. Many definitions and results are straightforward generalizations of their classical counterparts and have already been described in the literature. The definition of symplectic form in an $E$-manifold can be traced back to \cite{NestTsyganEMan}. The results concerning the formulation of Hamiltonian geometry and Marsden-Weinstein reduction over $E$-manifolds have been extracted from \cite{MarreroDeLeonAlgebroids, MarreroReduction, MarreroReductionGut}.

\begin{definition}
    Let $(M, E)$ be an $E$-manifold. An \emph{$E$-symplectic form} is a two-form $\omega \in \ee{\Omega}^2(M)$ which is non-degenerate and closed, that is, $\diff \omega = 0$.
\end{definition}

\begin{remark} \label{rmk:NonDegSymp}
    Here, non-degeneracy means that the vector bundle morphism
    \begin{equation}
        \begin{array}{rccc}
            \omega^\flat\colon & \ee{T}M & \longrightarrow & \ee{T}^*M \\
             & X & \longmapsto & \iota_X \omega
        \end{array}
    \end{equation}
    is a vector bundle isomorphism. The inverse map is denoted by $\omega^\sharp$ and, together, are called the \emph{musical isomorphisms} induced by $\omega$.
\end{remark}

The non-degeneracy condition can be used to define the notion of Hamiltonian vector field of a function in complete analogy with the smooth case.

\begin{definition}
    Let $(M, E)$ be an $E$-symplectic manifold. The \emph{Hamiltonian vector field} of a function $H \in \mathcal{C}^\infty(M)$, denoted by $X_H \in E$, is the unique $E$-field which satisfies
    \begin{equation}
        \iota_{X_H} \omega = - \diff H. \qedhere
    \end{equation}
\end{definition}

\subsection{The canonical symplectic form in $\ee{T}^* M$}\label{sec:canonical}

The Hamiltonian formalism in classical mechanics uses explicitly the Liouville one-form of the cotangent bundle $\mathrm{T}^* M$. There exists an analogue construction for the $E$-cotangent bundle.


Definition \ref{def:ProlFibreBun} of prolongation of a fibre bundle makes the natural projection $\tau\colon \ee{T}^*M \longrightarrow M$ an $E$-map. This condition is enough to define a Liouville form in the $E$-manifold $\ee{T}^*M$. This form gives a natural framework for Hamiltonian mechanics in $E$-manifolds by means of the canonical symplectic form, defined analogously to the smooth case. These constructions are particular cases of the Hamiltonian formalism for Lie algebroids as in \cite{MarreroDeLeonAlgebroids}.

\begin{definition}
    Let $\ee{T}M$ be an $E$-manifold and consider the pullback manifold $\ee{T} \ee{T}^* M$. The one-form $\lambda \in \ee{\Omega}^1(\ee{T}^* M)$, defined by its action on an element $X \in \ee{T}_\alpha (\ee{T}^*M)$ as
    \begin{equation} \label{eq:LiouvForm}
        \langle \lambda, X \rangle = \big \langle \tau_{\ee{T}\ee{T}^*M}(X), \diff_\alpha \tau_{\ee{T}^*M} (X) \big \rangle,
    \end{equation}
    is called the \emph{Liouville form} of $\ee{T}^* M$
\end{definition}

\begin{definition}
    Let $\ee{T}M$ be an $E$-manifold and consider the prolongation $\ee{T}\ee{T}^* M$. The \emph{canonical symplectic form} is defined in terms of the Liouville form as
    \begin{equation} \label{eq:CanSympForm}
        \omega = \diff \lambda. \qedhere
    \end{equation}
\end{definition}

We give now a local description of the Liouville form and the canonical symplectic form in the induced coordinates \eqref{prop:ENatCoord}. This result will be used to conclude that $\omega$ is non-degenerate and, consequently, a symplectic form.

\begin{lemma}
    Let $(M, E_M)$ be an $E$-manifold and consider the $E$-cotangent bundle $\ee{T}^*M$ with the pullback $E$-structure. Take an open set $U$, let $E_1, \ldots, E_p$ be a set of local generators of $\operatorname{Vec}_U(\ee{T}M)$ and consider the dual basis $F_1, \ldots, F_p$ in $\ee{T}^*M$. In natural coordinates \ref{prop:ENatCoord}, the Liouville one-form \eqref{eq:LiouvForm} is expressed as
    \begin{equation} \label{eq:ELiouFormNatCoord}
        \lambda = \sum_{i = 1}^n r_i E_i^*. \qedhere
    \end{equation}
\end{lemma}

\begin{proof}
    In natural coordinates \ref{prop:ENatCoord}, the bundle projections are given by $\tau_{\ee{T}\ee{T}^*M}(\vec{q}, \vec{p}, \vec{r}, \vec{s}) = (\vec{q}, \vec{r})$ and $\diff \tau_{\ee{T}^*M}(\vec{q}, \vec{p}, \vec{r}, \vec{s}) = (\vec{q}, \vec{p})$. Moreover, as $\{F_i\}$ is the dual basis of $\{E_i\}$, the natural pairing of $\ee{T}M$ and $\ee{T}^*M$ in these coordinates reads as $\langle \vec{r}, \vec{p} \rangle = \sum_{i = 1}^p r_i p_i$.
    
    Consider now an element $X \in \ee{T}_\alpha \ee{T}^*M$. By the definition of pullback bundle, there are scalars $p_i$ and $s_i$ such that
    \begin{equation*}
        X = \sum_{i = 1}^p p_i E_i + \sum_{i = 1}^p s_i V_i.
    \end{equation*}
    As a consequence of the local expressions for $\tau_{\ee{T}\ee{T}^*M}$ and $\diff \tau_{\ee{T}^*M}$, we have that
    \begin{equation*}
        \langle \lambda, X \rangle = \big\langle \tau_{\ee{T}\ee{T}^*M} (X), \diff \tau_{\ee{T}^*M}(X) \big\rangle = \langle \vec{r}, \vec{p} \rangle = \sum_{i = 1}^p r_i p_i.
    \end{equation*}
   Equation \eqref{eq:ELiouFormNatCoord} follows from this result.
\end{proof}

\begin{lemma}
    Let $(M, E_M)$ be an $E$-manifold and consider the $E$-cotangent bundle $\ee{T}^*M$ with the pullback structure. Take an open set $U$, let $E_1, \ldots, E_p$ be a set of local generators of $\operatorname{Vec}_U(\ee{T}M)$ and consider the dual basis $F_1, \ldots, F_p$ in $\ee{T}^*M$. In natural coordinates \ref{prop:ENatCoord}, the canonical symplectic form \eqref{eq:CanSympForm} is expressed as
    \begin{equation} \label{eq:ECanSympFormNatCoord}
        \omega = \sum_{i = 1}^p V_i^* \wedge E_i^* - \frac{1}{2} \sum_{i,j,k = 1}^p r_i C_{jk}^i E_j^* \wedge E_k^*.
    \end{equation}
    As a consequence, the symplectic form $\omega$ is non-degenerate.
\end{lemma}

\begin{proof}
    Let us consider the natural coordinates \ref{prop:ENatCoord}. Applying equations \eqref{eq:EProlFunDiff} and \eqref{eq:EProlSecDiff} to the local expression for the Liouville form \eqref{eq:ELiouFormNatCoord} we have
    \begin{equation*}
        \omega = \sum_{i = 1}^p \diff r_i \wedge E_i^* + \sum_{i = 1}^p r_i \diff V_i^* = \sum_{i = 1}^p V_i^* \wedge E_i^* - \frac{1}{2} \sum_{i,j,k = 1}^p r_i C_{jk}^i E_j^* \wedge E_k^*.
    \end{equation*}
    This expression implies the non-degeneracy of $\omega$. A straightforward computation shows that
    \begin{equation*}
        \frac{\omega^p}{p!} = V_1^* \wedge \cdots \wedge V_p^* \wedge E_1^* \wedge \cdots \wedge E_p^*
    \end{equation*}
    and, as a consequence, $\omega^p$ is a volume form.
\end{proof}

We will describe now the canonical symplectic forms associated to the examples \ref{ex:bMan}, \ref{ex:cMan} and \ref{ex:RegFol}.

\begin{example}
    Consider a $b$-manifold $(M, Z)$. In this setting, the $E$-tangent bundle $\ee{T}M$ is called the \emph{$b$-tangent bundle} and is written $\bi{T}M$. Similarly, the $E$-cotangent bundle is called the \emph{$b$-cotangent bundle}, $\bi{T}^* M$.
    The pullback structure on $\bi{T}M$ is equivalent to the $b$-structure of the induced hypersurface $Z_{\bi{T}M} \subset \bi{T}M$. This shows that definition \ref{def:ProlFibreBun} generalizes the $b$-manifold structure in the b-cotangent bundle, which can be found in \cite{GuilleminMirandaPires}.
    
    Take now an adapted chart $(U, \varphi)$ to $Z$ with coordinates $\vec{q}$; the set of sections
    \begin{equation*}
        q_1 \frac{\partial}{\partial q_1}, \frac{\partial}{\partial q_2}, \ldots, \frac{\partial}{\partial q_n}, \frac{\partial}{\partial p_1}, \ldots, \frac{\partial}{\partial p_n}.
    \end{equation*}
    is a basis of $\bi{T}M$. Notice that the structure functions $C_{ik}^k$ vanish. In the natural coordinates \ref{prop:ENatCoord}, the basis of sections $E_i^*, V_i^*$ can be explicitly described as
    \begin{equation*}
        \frac{\diff q_1}{q_1}, \ldots, \diff q_n, \diff p_1, \ldots, \diff p_n.
    \end{equation*}
    As a consequence of equation \eqref{eq:CanSympForm}, the canonical symplectic form is
    \begin{equation}
        \omega = \diff p_1 \wedge \frac{\diff q_1}{q_1} + \sum_{i = 2}^n \diff p_i \wedge \diff q_i.
    \end{equation}
    This is the canonical $b$-symplectic form in the $b$-cotangent bundle of \cite{GuilleminMirandaPires}.
\end{example}

\begin{example}
    Consider now, in the setting of example \ref{ex:cMan}, a $c$-manifold $(M, Z)$. Following the previous example, a set of local sections of a chart $(U, \varphi)$ around a point $p \in i(Z_k)\setminus i(Z_{k + 1})$ with coordinates $\vec{q}$ adapted to the immersed hypersurface $i(Z)$ induces the basis of local sections $E_i^*$, $V_i^*$ as
    \begin{equation*}
        \frac{\diff q_1}{q_1}, \ldots, \frac{\diff q_k}{q_k}, \diff q_{k + 1}, \ldots, \diff q_n, \diff p_1, \ldots, \diff p_n. 
    \end{equation*}
    As before, the structure functions $C_{ij}^k$ vanish. Equation \eqref{eq:CanSympForm} implies that the canonical Liouville form is expressed as
    \begin{equation}
        \omega = \sum_{i = 1}^k \diff p_i \wedge \frac{\diff q_i}{q_i} + \sum_{i = k + 1}^n \diff p_i \wedge \diff q_i.
    \end{equation}
    
    We remark that this construction is canonical in the sense that it can be canonically constructed in the $E$-cotangent bundle $\ee{T}^*M$, in analogy with the smooth canonical symplectic form. Miranda and Scott proved in \cite{MirandaScott} that there is no analogue of Darboux's theorem for symplectic forms over $E$-manifolds; rather, there are some cohomological obstructions to the existence of normal forms.
\end{example}

\begin{example}
    In the setting of example \ref{ex:RegFol}, consider a foliated chart $(U, \varphi)$ with coordinates $\vec{q}$. The basis of sections $E_i^*$, $V_i^*$ in \ref{prop:ENatCoord} are locally described as $\diff q_1, \ldots, \diff q_k, \diff p_1, \ldots, \diff p_k$ and, as the structure functions vanish once again, equation \eqref{eq:CanSympForm} is read as
    \begin{equation}
        \omega = \sum_{i = 1}^k \diff p_i \wedge \diff q_i. \qedhere
    \end{equation}
\end{example}

\subsection{Lifting $E$-maps, cotangent lifts and symplectic properties}

The definition of $E$-map consists of two pieces of data; one is the map between the base manifolds, while the other is the covering map between the corresponding $E$-manifolds (in fact, the covering map fully determines the base map). We could wonder if, similar to the smooth situation, the base map completely determines the covering map via a tangent functor. The answer is negative, as the following example shows.

\begin{example}
    Let us consider a $b$-manifold $(M, Z)$ and $p_0 \in Z$. Define the map $f\colon M \longrightarrow M$ as $f(p) = p_0$ for every $p \in M$. Obviously, $\diff_p f (\bi{\operatorname{Vec}} (p)) = 0 \subset \bi{\operatorname{Vec}}(p_0)$, and the map admits a restriction $F\colon \bi{T} M \longrightarrow \rho_E(\bi{T}M)$, but there is no unique way to lift this map to a map $G\colon \bi{T}M \longrightarrow \bi{T}M$. In fact, for any $X \in \bi{T}_{p_0} M$ the map $G\colon \bi{T}M \longrightarrow \bi{T}M$ defined by $G(Y) = X$ for every $Y \in \bi{T}M$ makes diagrams \eqref{eq:EMapDiag} commutative.
\end{example}

The key point of the previous counterexample is that the anchor $\rho$ of an $E$-manifold only identifies $\ee{T}M$ with $\mathrm{T}M$ in an open and dense subset of $\mathrm{T}M$. We are defining $f$ such that $\diff f$ does not intersect the isomorphism locus of $\rho_E$. However, there is a concrete instance where the map $f$ uniquely lifts to an $E$-map. This class of maps is remarkably important for us, because we can define an analogous version of the cotangent lift.

\begin{proposition} \label{prop:CompDiffDetEMaps}
    Let $(M, E_M)$ and let $f\colon M \longrightarrow M$ be a smooth diffeomorphism such that $\diff_p f (E_M(p)) = E_M(f(p))$ for every $p \in M$. Then, there exists a unique $E$-map covering $f$, which we denote by $\diff f$ and call the \emph{$E$-tangent lift} of $f$.
\end{proposition}

\begin{proof}
    From the condition $\diff_p f (E_M(p)) = E_M(f(p))$ for every $p \in M$ we have that there exists a well-defined map $F\colon \ee{T}M \longrightarrow \rho_E(\ee{T}M)$. Let $U \subseteq M$ be the set of points for which $\rho_E$ is an isomorphism, which is open and dense from the construction of the $E$-tangent bundle. As $f$ is a diffeomorphism, $f(U)$ is also open and dense and $U \cap f(U)$ is open and dense. Therefore, over $U \cap f(U)$, $F$ uniquely lifts to a map $\diff f\colon \ee{T}M \longrightarrow \ee{T}M$ which, in fact, is a diffeomorphism because $\rho_E$ and $F$ are diffeomorphisms as well. Finally, as $U \cap f(U)$ is open and dense, $\diff f$ can be smoothly extended to a diffeomorphism $\diff\colon \ee{T}M \longrightarrow \ee{T}M$ over $M$.
\end{proof}

As structure preserving diffeomorphisms in the base manifold uniquely lift to $E$-maps, it is natural to consider lifts to the $E$-cotangent bundle. The following proposition shows that such lifts are well-defined, $E$-maps with respect to the pullback structure in $\ee{T}^*M$ and that the canonical Liouville form \eqref{eq:LiouvForm} is preserved.

\begin{proposition} \label{prop:ECotLiftPresLiouForm}
    Let $(M, E_M)$ be an $E$-manifold and consider the $E$-cotangent bundle $\ee{T}^*M$ with the pullback structure. Any $E$-diffeomorphism $f \in \operatorname{Diff}(\ee{T}M)$ defines an $E$-diffeomorphism $\hat{f} \in \operatorname{Diff}(\ee{T}\ee{T}^*M)$, called the \emph{$E$-cotangent lift}. Additionally, $\hat{f}$ preserves the Liouville one-form \eqref{eq:LiouvForm}.
\end{proposition}

\begin{proof}
    Consider a diffeomorphism $f\colon M \longrightarrow M$ compatible with $E_M$ which we know, by proposition \ref{prop:CompDiffDetEMaps}, that uniquely lifts to a covering diffeomorphism $\diff f\colon \ee{T}M \longrightarrow \ee{T}M$. We define, similarly to the smooth case, the $E$-cotangent lift as $\hat{f} = (\diff f^{-1})^*$. This map satisfies $\widehat{fg} = \hat{f} \hat{g}$ and gives rise to the commutative diagram \ref{diag:ECotLift}.
    \begin{equation} \label{diag:ECotLift}
    \begin{tikzcd}[sep = large]
    	{\ee{T}^*M} & {\ee{T}^*M} \\
    	M & M
    	\arrow["f"', from=2-1, to=2-2]
    	\arrow["{\hat{f}}", from=1-1, to=1-2]
    	\arrow["\tau"', from=1-1, to=2-1]
    	\arrow["\tau", from=1-2, to=2-2]
    \end{tikzcd}
    \end{equation}
    
    As the cotangent lift is a diffeomorphism, checking that it is an $E$-map amounts to checking that it preserves the pullback structure by proposition \ref{prop:CompDiffDetEMaps}. From the commutativity of diagram \eqref{diag:ECotLift}, we can see that $\diff \tau \diff \hat{f} (E_{\ee{T}^*M}) = \diff f \diff \tau(E_{\ee{T}^*M})$. Given that $E_{\ee{T}^*M}= \langle (\diff \tau)^{-1}(E_M) \rangle$ and that $f$ is an $E$-map, we have $\diff f \diff \tau(E_{\ee{T}^*M}) = \diff f(E_M) = E_M$. As a result, $\diff \tau \diff \hat{f} (E_{\ee{T}^*M}) = E_M$, which implies $\diff \hat{f}(E_{\ee{T}^*M}) \subset E_{\ee{T}^*M}$. As $\hat{f}$ is a diffeomorphism, $\diff \hat{f}(E_{\ee{T}^*M}) \subset E_{\ee{T}^*M}$ and the proof is finished. 
    
    The proof of the second statement follows from the commutativity of the following diagrams:
    \begin{equation} \label{diag:ECotLiftTan}
    \begin{tikzcd}[sep = large]
    	{\ee{T}\ee{T}^*M} & {\ee{T}\ee{T}^*M} \\
    	{\ee{T}M} & {\ee{T}M}
    	\arrow["{\diff f}"', from=2-1, to=2-2]
    	\arrow["{\diff \hat{f}}", from=1-1, to=1-2]
    	\arrow["{\diff \tau_{\ee{T}^*M}}", from=1-2, to=2-2]
    	\arrow["{\diff \tau_{\ee{T}^*M}}"', from=1-1, to=2-1]
    \end{tikzcd} \quad
    \begin{tikzcd}[sep = large]
    	{\ee{T}\ee{T}^*M} & {\ee{T}\ee{T}^*M} \\
    	{\ee{T}^*M} & {\ee{T}^*M}
    	\arrow["{\hat{f}}"', from=2-1, to=2-2]
    	\arrow["{\diff \hat{f}}", from=1-1, to=1-2]
    	\arrow["{\tau_{\ee{T}\ee{T}^*M}}", from=1-2, to=2-2]
    	\arrow["{\tau_{\ee{T}\ee{T}^*M}}"', from=1-1, to=2-1]
    \end{tikzcd}
    \end{equation}
    
    A direct computation shows
    \begin{align*}
        \langle \hat{f}^*\lambda, X \rangle &= \langle \lambda, \diff \hat{f} X \rangle \\
        &= \big\langle \tau_{\ee{T}\ee{T}^*M}(\diff \hat{f} X), \diff \tau_{\ee{T}^*M} \diff \hat{f} X \big\rangle \\
        &= \big\langle \hat{f} \tau_{\ee{T}\ee{T}^*M}(X), \diff f \diff \tau_{\ee{T}^*M} X \big\rangle \\
        &= \big\langle \tau_{\ee{T}\ee{T}^*M}(X), \diff \tau_{\ee{T}^*M} X \big\rangle \\
        &= \langle \lambda, X \rangle. \qedhere
    \end{align*}
\end{proof}

\subsection{$E$-group actions and Marsden-Weinstein reduction}

A key technique in classical Hamiltonian mechanics is that of Marsden-Weinstein reduction, which formalizes the idea of reducing the phase space of a physical system by a symmetry. This idea is also completely necessary in the formulation of classical gauge and Yang-Mills theories, which further motivates their study. Underlying this construction is the notion of Lie group action on a symplectic manifold, which encodes the information of a continuous symmetry and infinitesimal action. 

\begin{definition}
    Let $(M, E)$ be an $E$-manifold and consider a Lie group $G$. A \emph{Lie group action} of $G$ on $(M, E_M)$ is Lie group action $\rho\colon G \times M \longrightarrow M$ such that, for every $(g, p) \in G \times M$, we have $\diff_{(g, p)} \rho (E_{G \times M}) = E_M$. Here, we have taken the $E$-manifold structure $(G \times M, \operatorname{Vec}(G) \oplus E_M)$.
\end{definition}



For a fixed $g \in G$, the group action defines a diffeomorphism in $M$ and, as a consequence of the previous definition, the map is lifted to the $E$-tangent bundle. Therefore, this information is equivalent to defining Lie group actions in a categorical sense. 

\begin{remark}
    The definition of $E$-action implies that the fundamental vector fields of $\rho$ are $E$-fields. The fundamental action $\hat{\rho}$ can be recovered from the action $\rho$ as $\hat{\rho}(X) = \diff \rho (X \oplus \vec{0})$ for any left-invariant vector field $X \in \operatorname{Lie}(G) \simeq \mathfrak{g}$. From the universal property of products \ref{prop:UnivPropEProd}, the pushforward is an $E$-field, $\diff \rho(X \oplus \vec{0}) \in \ee{\operatorname{Vec}}(M)$.
\end{remark}

The other ingredient in the formulation of Marsden-Weinstein reduction is that of \emph{Hamiltonian} group actions. The definition in the case of $E$-manifolds follows directly from that in the case of smooth manifolds. We also state a technical lemma concerning the existence of product Hamiltonian group actions without proof.

\begin{definition}[Hamiltonian action]
    Let $(M, E)$ be an $E$-manifold and consider a symplectic form $\omega \in \ee{\Omega}^2(M)$. We say that an action $\rho\colon G \times M \longrightarrow M$ is \emph{Hamiltonian} if there exists a moment map $\mu \in \mathcal{C}^\infty(M) \otimes \mathfrak{g}^*$ fulfilling
    \begin{equation}
        \iota_{X^\sharp} \omega = - \diff \langle \mu, X \rangle
    \end{equation}
    for each $X \in \mathfrak{g}$.
\end{definition}


\begin{lemma} \label{lem:ProdHamAct}
    Let $(M, E_M)$ and $(N, E_N)$ be two $E$-symplectic manifolds with symplectic forms $\omega_M$ and $\omega_N$, respectively and let $G$ be a Lie group. Given two Hamiltonian $G$-actions $\rho_M\colon G \times M \longrightarrow M$ and $\rho_N\colon G \times N \longrightarrow N$ with respective moment maps $\mu_M$ and $\mu_N$, the induced $G$-action in $M \times N$ with symplectic form $p_1^* \omega_M + p_2^* \omega_N$ is Hamiltonian with moment map $\mu = p_1^* \mu_M + p_2^* \mu_N$.
\end{lemma}

    

In the definition of Hamiltonian action we have described the fundamental vector fields of the action as Hamiltonian vector fields generated by smooth functions $\mathcal{C}^\infty(M)$. This construction is well defined because the sheaf of smooth functions fits into the cochain complex of $E$-differential forms
\begin{equation} \label{eq:ECochainCplx}
    \begin{tikzcd}
    	0 & {\mathcal{C}^\infty(M)} & {\ee{\Omega}^1(M)} & \cdots & {\ee{\Omega}^n(M)} & 0
    	\arrow[from=1-1, to=1-2]
    	\arrow[from=1-2, to=1-3]
    	\arrow[from=1-3, to=1-4]
    	\arrow[from=1-4, to=1-5]
    	\arrow[from=1-5, to=1-6]
    \end{tikzcd}
\end{equation}
However, in some instances the class of smooth functions may be enlarged to consider more general modules of functions which include $\mathcal{C}^\infty(M)$ as a submodule and still fit into equation \eqref{eq:ECochainCplx}. The first example to consider are $b$-functions, as described in \cite[def. 3]{GMPToricActions}. 
We define $E$-functions in the following way: Let $\alpha$ be a $E$-form of degree $1$ which is closed. Obviously this form is not the differential of a smooth function but we enlarge the set of admissible functions in such a way that a closed form is exact in this category. In other words, as observed by Nest and Tsygan in 
 \cite{NestTsyganEMan}, in distinction  to the standard de Rham complex, the $E$-complex is
not locally acyclic and hence does not give an acyclic resolution of the set of smooth functions. However, we can still define \emph{artificially} a set such that this condition is met. We call this, the set of $E$-functions. The exact description of this set depends on the example under consideration. For $b$-functions this set is exactly $^{b}\mathcal{C}^\infty(M)=\{ g \log\vert x\vert+ h, g \in \mathbb{R}, h\in \mathcal{C}^\infty(M)\}.$ The set of $b^m$-functions is defined recursively according to the formula $$~^{b^m} \mathcal{C}^\infty(M)= x^{-(m-1)}\mathcal{C}^\infty(x) + ~^{b^{m-1}} \mathcal{C}^\infty(M)$$

\noindent  {with $\mathcal{C}^\infty(x)$ the set of smooth functions in the defining function $x$}.

In the case of regular foliations this set coincides with the set of $0$-forms in foliated cohomology, thus with basic functions.

The notion of Hamiltonian group action can be generalized to include Hamiltonian functions which belong to the set $\ee{\mathcal{C}}^\infty(M)$. This definition, at least for $b^m$-symplectic manifolds, can be found in \cite[def. 2.6]{MatveevaMiranda}.

\begin{definition}[$E$-Hamiltonian action]
    Let $(M, E)$ be an $E$-manifold and consider a symplectic form $\omega \in \ee{\Omega}^2(M)$. We say that an action $\rho\colon G \times M \longrightarrow M$ is \emph{$E$-Hamiltonian} if there exists a moment map $\mu \in \ee{\mathcal{C}}^\infty(M) \otimes \mathfrak{g}^*$ fulfilling
    \begin{equation}
        \iota_{X^\sharp} \omega = - \diff \langle \mu, X \rangle
    \end{equation}
    for each $X \in \mathfrak{g}$.
\end{definition}


We are ready now to state a version of the Marsden-Weinstein reduction by the Hamiltonian action of a Lie group $G$ in the context of $E$-symplectic manifolds. This result has been extracted and adapted from \cite[theorem 3.11]{MarreroReductionGut}. Afterwards, we prove a version of the well-known shifting trick over $E$-symplectic manifolds. The proof is essentially the same as for regular manifolds, and all we have to do is check that the construction is compatible with the $E$-structures.

\begin{theorem} \label{thm:ERed}
    Let $(M, E)$ be an $E$-manifold with symplectic form $\omega \in \ee{\Omega}^2(M)$. Consider a proper and free group $E$-action $\rho\colon G \times M \longrightarrow M$ which is Hamiltonian with moment map $\mu\colon M \longrightarrow \mathfrak{g}^*$. If $\alpha \in \mathfrak{g}^*$ is a regular value of $\mu$, $\mu^{-1}(\alpha)/G_\alpha$ is an $E$-symplectic manifold with symplectic form $\omega_\mathrm{red}$ given by
    \begin{equation} \label{eq:RedForm}
        \pi^* \omega_\mathrm{red} = i^* \omega. \qedhere
    \end{equation}
\end{theorem}

\begin{theorem}[$E$-Shifting trick] \label{thm:EShiftTrick}
    Let $(M, E)$ be an $E$-symplectic manifold with symplectic form $\omega \in \ee{\Omega}^2(M)$ and a free and proper $E$-Lie group action $\rho\colon G \times M \longrightarrow M$ with moment map $\mu\colon M \longrightarrow \mathfrak{g}^*$ and assume $\alpha \in \mathfrak{g}^*$ is a regular value of $\mu$. If we endow $(\mathcal{O}(-\alpha), \omega_{\mathcal{O}(- \alpha)})$ with the natural coadjoint action of $G$ and the moment map $i \colon \mathcal{O}(-\alpha) \longhookrightarrow \mathfrak{g}^*$, there is an isomorphism of $E$-symplectic manifolds
    \begin{equation} \label{eq:ShiftTrick}
        \mu^{-1}(\alpha)/G_\alpha \simeq (\mu + i)^{-1}(0)/G.
    \end{equation}
    In particular, the Marsden-Weinstein reduction can always be performed at $0 \in \mathfrak{g}^*$.
\end{theorem}

\begin{proof}
    The first part of the proof relies on an interesting observation on its own, which claims that the reduction by coadjoint orbits is equivalent to the reduction by a point. Given a regular value $\alpha \in \mathfrak{g}^*$, the natural inclusion of $E$-manifolds $i\colon \mu^{-1}(\alpha) \longhookrightarrow \mu^{-1}(\mathcal{O}(\alpha))$ induces an $E$-diffeomorphism $\mu^{-1}(\alpha)/G_\alpha \simeq \mu^{-1}(\mathcal{O}(\alpha))/G$. We will also show that the uniqueness of the reduced form following equation \eqref{eq:RedForm} implies that the map of classes $[i]\colon \mu^{-1}(\alpha)/G_\alpha \longrightarrow \mu^{-1}(\mathcal{O}(\alpha))/G$ is an $E$-symplectomorphism. 
    
    Let now $\pi_\alpha\colon \mu^{-1}(\alpha)/G_\alpha \longrightarrow \mu^{-1}(\alpha)/G_\alpha$ and $\pi_{\mathcal{O}(\alpha)}\colon \mu^{-1}(\mathcal{O}(\alpha)) \longrightarrow \mu^{-1}(\mathcal{O}(\alpha))/G$ be the respective quotient maps; similarly, let $i_\alpha\colon \mu^{-1}(\alpha) \longhookrightarrow M$ and $i_{\mathcal{O}(\alpha)}\colon \mu^{-1}(\mathcal{O}(\alpha)) \longhookrightarrow M$ the respective inclusions and let $\omega_{\alpha}$ and $\omega_{\mathcal{O}(\alpha)}$ be the reduced symplectic forms. It is obvious that $i_\alpha = i_{\mathcal{O}(\alpha)} i$, which implies $i_\alpha^* \omega = i^* i_{\mathcal{O}(\alpha)}^* \omega$. The characterization of the reduced $E$-symplectic form following equation \eqref{eq:RedForm} implies $\pi_\alpha^* \omega_\alpha = i^* \pi_{\mathcal{O}(\alpha)}^* \omega_{\mathcal{O}(\alpha)}$. We construct now the commutative diagram \eqref{eq:UniqRedForm}, where all the maps are actually $E$-maps. A direct consequence of the commutativity is that $\pi_\alpha^* [i]^* \omega_{\mathcal{O}(\alpha)} = i^* \pi_{\mathcal{O}(\alpha)}^* \omega_{\mathcal{O}(\alpha)}$. The previous results imply now that $\pi_\alpha^* \omega_\alpha = \pi_\alpha^* [i]^* \omega_{\mathcal{O}(\alpha)}$ but, as the reduced form is unique, we conclude $\omega_\alpha = [i]^* \omega_{\mathcal{O}(\alpha)}$. This shows that $[i]$ is an $E$-symplectomorphism.
    \begin{equation} \label{eq:UniqRedForm}
        \begin{tikzcd}
        	{\mu^{-1}(\alpha)} & {\mu^{-1}(\mathcal{O}(\alpha))} \\
        	{\mu^{-1}(\alpha)/G_\alpha} & {\mu^{-1}(\mathcal{O}(\alpha))/G}
        	\arrow["{[i]}"', from=2-1, to=2-2]
        	\arrow["i", from=1-1, to=1-2]
        	\arrow["{\pi_\alpha}"', from=1-1, to=2-1]
        	\arrow["{\pi_{\mathcal{O}(\alpha)}}", from=1-2, to=2-2]
        \end{tikzcd}
    \end{equation}
    
    After this isomorphism of $E$-manifolds, we consider the product manifold $M \times \mathcal{O}(-\alpha)$, which is endowed with the product $E$-symplectic form of the $E$-form $\omega$ and the symplectic form on the coadjoint orbit: $\omega \oplus \omega_{\mathcal{O}(-\alpha)}$. The diagonal action of $G$ is Hamiltonian with moment map $\mu + i$ following lemma \ref{lem:ProdHamAct}. A straightforward computation shows
    \begin{align*}
        (\mu + i)^{-1}(0) &= \{ (p, \beta) \in M \times \mathcal{O}(- \alpha) \mid \mu(p) + \beta = 0 \} \\
        &= \mu^{-1}(\mathcal{O}(\alpha)),
    \end{align*}
    and, in particular, the isomorphism respects the corresponding $E$-structures. This observation, together with the first result, proves equation \eqref{eq:ShiftTrick}.
\end{proof}

\section{Gauge theory of $E$-manifolds} \label{sec:GaugeTheory}

In classical gauge theories, the configuration space of a point-mass particle in terms of positions and momenta does not give a complete characterization of the physical properties of a system. A standard example is the configuration space of classical electromagnetism, in which particles are described by an additional scalar called the \emph{electric charge}. The natural configuration space becomes a real line bundle $\mathbf{L}$ over the space-time $M$. In a general setting, the new configuration space is described by a fibre bundle $\tau\colon B \longrightarrow M$ with typical fibre $Q$.  This fibre is understood to represent the internal degrees of freedom of the particle.

Additionally, in classical gauge theories we consider the action of a Lie group $G$ on the total configuration space by gauge transformations. Following the previous interpretation, gauge transformations represent diffeomorphisms of the total configuration space $B$ which do not change the physical description of a particle in the base manifold $M$. The geometric idea behind a gauge transformation is represented in picture \ref{fig:gauge-local}.

Having recalled the geometric description of gauge theories, we start by defining gauge theories over $E$-manifolds using the pullback structure over a fibre bundle.

\begin{definition}
    Let $(M, E_M)$ be an $E$-manifold. A \emph{gauge theory} over $M$ is a pair $(B, G)$, where $\tau\colon B \longrightarrow M$ is a fibre bundle with the pullback structure $E_B$ and a $G$-structure on $B$. 
\end{definition}

\begin{figure}[h!]
    \centering
    \includegraphics[scale=0.5]{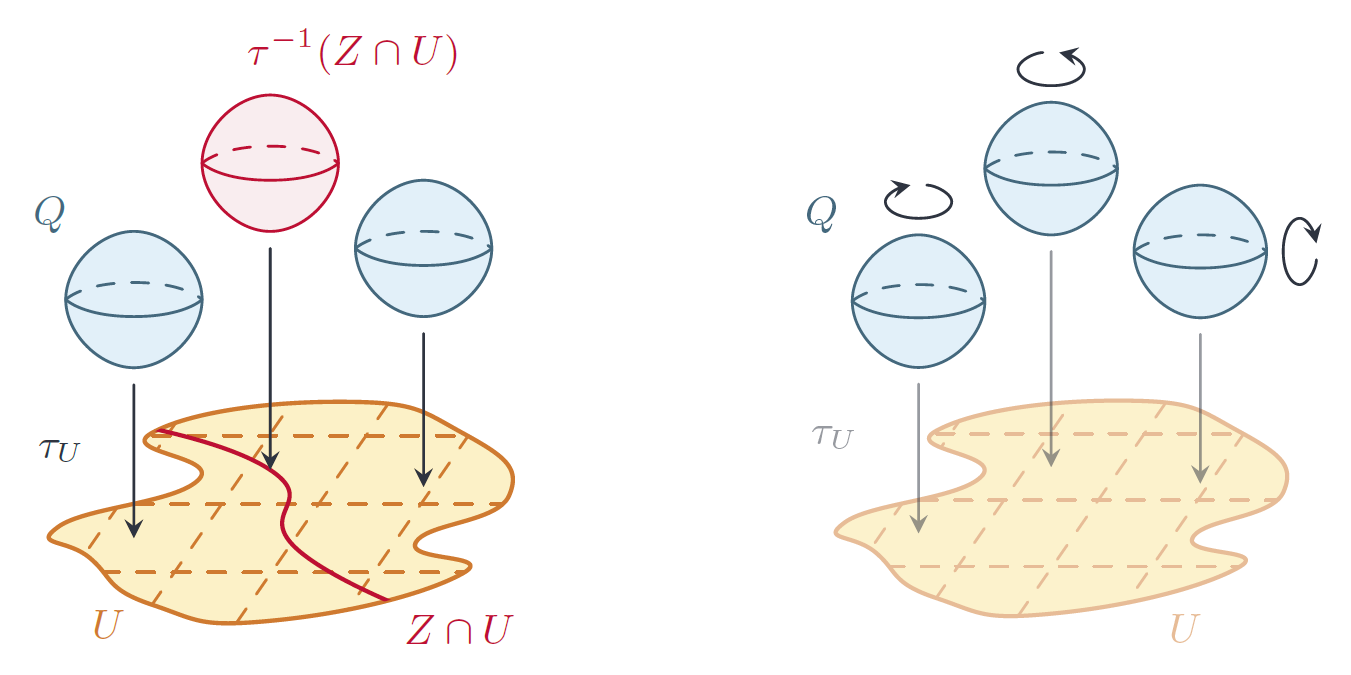}
    \caption{On the left, the induced pullback structure described locally for a $b$-manifold. In this case, the $b$-structure coincides with that generated by the pullback hypersurface $\pi^{-1}(Z)$. On the right, a schematic representation of a gauge transformation. The transformation may change from fibre to fibre smoothly, but a point of a fibre has to be mapped to the same fibre.}
    \label{fig:gauge-local}
\end{figure}

The equations of motion are specified using a $G$-invariant connection on the bundle $B$, called a gauge field. Such connections were already described in \cite{NestTsyganEMan}; we give here a different characterization as invariant splittings of an ``$E$-Atiyah sequence''.

\begin{definition}
    Let $(M, E_M)$ be an $E$-manifold and consider a gauge theory $(B, G)$ over $M$. A \emph{gauge field} on $(B, G)$ is a $G$-invariant splitting of the following short exact sequence:
    \begin{equation} \label{eq:EAtiyahSeq}
        \begin{tikzcd}
        	0 & {\ker \diff \tau} & {\ee{T} B} & {\tau^* \ee{T} M}  & 0.
        	\arrow[from=1-1, to=1-2]
        	\arrow["\iota", from=1-2, to=1-3]
        	\arrow["{\diff \tau}", from=1-3, to=1-4]
        	\arrow[from=1-4, to=1-5]
        \end{tikzcd} \qedhere
    \end{equation}
\end{definition}

The equations of motion of a classical particle under the action of a gauge field were described by  Wong in  \cite{Wong}, and are named after him as \emph{Wong's equations}. Weinstein showed in \cite{WeinsteinUniversal} that the equations of motion by Wong can be described in geometric terms and that, moreover, they become Hamiltonian. In his setting, a gauge field induces a projection from the cotangent bundle $\mathrm{T}^*B$, which can be regarded as the natural phase space of a gauge theory, to the standard cotangent bundle $\mathrm{T}^*M$. The choice of a Hamiltonian function in $\mathrm{T}^*M$ (the kinetic energy, for instance, if $M$ is a Riemannian manifold) induces a Hamiltonian function in $\mathrm{T}^*B$ which generates Wong's equations. 

Sternberg introduced generalized the minimal coupling procedure of electromagnetism to Yang-Mills theories and showed that it could be described by the introduction of a magnetic term in the canonical symplectic structure of $\mathrm{T}^*M$. Weinstein showed that the choice of a gauge field induces a symplectomorphism of symplectic manifolds which describes Sternberg's minimal coupling. Montgomery considered the symplectic formulation of Yang-Mills theories and showed that the minimal coupling of Weinstein and Wong's equations of motion could be described over more general Poisson manifolds. In particular, he proved that the restriction to a symplectic leaf recovers Weinstein's results.

Our goal throughout the rest of the section is to prove analogous statements when the base manifold $M$ has an $E$-structure.

\subsection{Geodesics on Riemannian $E$-manifolds}

In classical mechanics, the base manifold of a physical system is called the \emph{configuration space}, and is assumed to be a Riemannian manifold. The study of geodesics on a Riemannian manifold is one of the first approaches to general gauge theories in physics. Metrics can be similarly defined for arbitrary $E$-manifolds, and their construction is analogue to that of $E$-symplectic forms.

\begin{definition}
    Consider an $E$-manifold $(M, E)$. A \emph{metric} on $(M, E)$ is a section $g \in \Gamma(\ee{T}M \odot \ee{T}M)$ which is positive-definite and non-degenerate. An $E$-manifold $\ee{T}M$ together with a metric $g$ is called a \emph{Riemannian $E$-manifold}.
\end{definition}

Analogously to remark \ref{rmk:NonDegSymp}, the non-degeneracy of the metric $g$ gives rise to vector bundle isomorphisms $g^\flat, g^\sharp$, called the musical isomorphisms of the metric. 
Affine connections are defined as splittings of the sequence
\begin{equation}
    \begin{tikzcd}
    	0 & {\ker \diff \tau} & {\ee{T} \ee{T} M} & {\ee{T}M} & 0.
    	\arrow[from=1-1, to=1-2]
    	\arrow["\iota", from=1-2, to=1-3]
    	\arrow["{\diff \tau}", from=1-3, to=1-4]
    	\arrow[from=1-4, to=1-5]
    \end{tikzcd}
\end{equation}
The lift $X_\nabla \colon \ee{T}M \longrightarrow \ee{T}\ee{T}M$ is an $E$-field in $\ee{T}M$ with the induced $E$-structure on $\ee{T}M$, called the \emph{geodesic flow}.  To construct such a flow in the dual bundle $\ee{T}^*M$ we consider the musical isomorphism given by the metric $g$, giving the following isomorphism of short exact sequences:
\begin{equation}
    \begin{tikzcd}[sep=large]
    	0 & {\ker\diff\tau} & {\ee{T}\ee{T}M} & {\ee{T}M} & 0 \\
    	0 & {\ker\diff\tau} & {\ee{T}\ee{T}^* M} & {\ee{T}^* M} & 0
    	\arrow[from=1-1, to=1-2]
    	\arrow[from=2-1, to=2-2]
    	\arrow["\iota", from=1-2, to=1-3]
    	\arrow["\iota"', from=2-2, to=2-3]
    	\arrow["{\diff\tau_{\ee{T}M}}", shift left=1, from=1-3, to=1-4]
    	\arrow["{\diff\tau_{\ee{T}^*M}}"', shift right=1, from=2-3, to=2-4]
    	\arrow[from=1-4, to=1-5]
    	\arrow[from=2-4, to=2-5]
    	\arrow["{\diff g^\flat}", shift left=0, from=1-2, to=2-2]
    	\arrow["{\diff g^\flat}", shift left=0, from=1-3, to=2-3]
    	\arrow["{g^\flat}", shift left=0, from=1-4, to=2-4]
    	\arrow["X_\nabla", shift left=1, from=1-4, to=1-3]
    	\arrow["Y_\nabla"', shift right=1, from=2-4, to=2-3]
    \end{tikzcd}
\end{equation}

The geodesics are the flow of a Hamiltonian vector fields also for $E$-manifolds. 

\begin{theorem} \label{thm:EGeod}
    Consider an $E$-manifold $(M, E)$ with a metric $g \in \Gamma(\ee{T}^*M \odot \ee{T}^*M)$ and an affine connection $X_\nabla\colon \ee{T}^* M \longrightarrow \ee{T}\ee{T}^* M$. The induced flow $\varphi_t$ of $Y_\nabla = \diff g^\flat X_\nabla$ is Hamiltonian, with Hamiltonian function
    \begin{equation*}
        H(\alpha) = g(g^\sharp \alpha, g^\sharp \alpha). \qedhere
    \end{equation*}
\end{theorem}

\begin{proof}
    For the first part of the proof we only have to observe that the flow $\varphi_t$ gives an isomorphism of short exact sequences
    \begin{equation*}
        \begin{tikzcd}[sep=large]
        	0 & {\ker\diff\tau} & {\ee{T}\ee{T}^*M} & {\ee{T}^*M} & 0 \\
        	0 & {\ker\diff\tau} & {\ee{T}\ee{T}^* M} & {\ee{T}^* M} & 0
        	\arrow[from=1-1, to=1-2]
        	\arrow[from=2-1, to=2-2]
        	\arrow["\iota", from=1-2, to=1-3]
        	\arrow["\iota"', from=2-2, to=2-3]
        	\arrow["{\diff\tau_{\ee{T}^*M}}", shift left=1, from=1-3, to=1-4]
        	\arrow["{\diff\tau_{\ee{T}^*M}}"', shift right=1, from=2-3, to=2-4]
        	\arrow[from=1-4, to=1-5]
        	\arrow[from=2-4, to=2-5]
        	\arrow["{\diff\varphi_t}", shift left=0, from=1-3, to=2-3]
        	\arrow["{\varphi_t}", shift left=0, from=1-4, to=2-4]
        	\arrow["Y_\nabla", shift left=1, from=1-4, to=1-3]
        	\arrow["Y_\nabla"', shift right=1, from=2-4, to=2-3]
        	\arrow[from=1-2, to=2-2]
        \end{tikzcd}
    \end{equation*}
    Notice that the splitting $Y_\nabla$ remains unchanged under the action of $\varphi_t$ as $[Y_\nabla, Y_\nabla] = 0$, and, as a consequence, $\diff\varphi_t Y_\nabla = Y_\nabla$. We can conclude that the flow $\varphi_t$ preserves the Liouville form in $\ee{T}^*M$ using a similar argument to that in the proof of proposition \ref{prop:ECotLiftPresLiouForm}. 
    Thus, $\mathcal{L}_{Y_\nabla} \lambda = 0$ and because of Cartan's formula:
    \begin{equation*}
        \iota_{Y_\nabla} \diff \lambda = - \diff \iota_{Y_\nabla} \lambda,
    \end{equation*}
    thus, proving that $Y_\nabla$ is Hamiltonian.
    
    Regarding the Hamiltonian function, from the previous result  $H(\alpha) = \langle \lambda_\alpha, Y_\nabla(\alpha) \rangle$, thus, by definition \eqref{eq:LiouvForm}, equals $\langle \alpha, \diff \tau_{\ee{T}^*M} Y_\nabla(\alpha) \rangle$. Take now $X = g^\sharp(\alpha)$ or, equivalently, $\alpha = g^\flat(X)$. Given that the musical isomorphisms are vector bundle morphisms, we obtain $\diff\tau_{\ee{T}^* M} \diff g^\flat = \diff \tau_{\ee{T}M}$. Recall that $Y_\nabla = \diff g^\flat X_\nabla$. As a consequence,
    \begin{align*}
        \langle \lambda_\alpha, Y_\nabla(\alpha) \rangle &= \langle \alpha, \diff \tau_{\ee{T}^* M} \diff g^\flat X_\nabla(X)) \rangle \\
        &= \langle \alpha, \diff \tau_{\ee{T}M} X_\nabla(X) \rangle \\
        &= \langle \alpha, X \rangle \\
        &= g(X, X).
    \end{align*}
    Therefore, we conclude that the Hamiltonian function of the geodesic flow $\varphi_t$ is $g(g^\sharp \alpha, g^\sharp \alpha)$.
\end{proof}

\subsection{Principal bundles and gauge fields over $E$-manifolds}

The standard presentation of $E$-gauge theories introduced in the beginning of section \ref{sec:GaugeTheory} can be recovered from a formulation in terms of \emph{associated} principal bundles. The symplectic formulation of the equations of motion for gauge theories is commonly stated in these terms (see for instance, \cite{WeinsteinUniversal, MontgomeryWongsEquations, MontgomeryThesis}).

This subsection is devoted to the study of the pullback structure in principal $G$-bundles over an $E$-manifold $M$, as well as the characterizations of gauge fields in this setting. These properties will be used throughout the study of the symplectic formulation of the equations of motion in the following subsections. We begin with an example of how gauge theories are formulated in the context of manifolds admitting an $E$-symplectic structure and then we present a lemma on the structure of the vertical bundle of an $E$-principal $G$-bundle.

\begin{example}[Relativistic electromagnetism on $E$-manifolds]
In example \ref{example:Minkowskispace} we saw that, in a Minkowski space $(M,g)$, the singular foliation of $\mathrm{T}M$ by level sets $\mathcal{F}_k$ of the kinetic energy is the $E$-manifold $\mathrm{T}\mathcal{F}$. In the setting of special relativity, the space-time is $\mathbf{R}^4$ with the flat metric and it admits a relativistic version of the Maxwell’s equations of electromagnetism.

To introduce an electromagnetism formulation that behaves well with the Lorentz transformations in special relativity, instead of talking separately about \emph{charge density} $q$ and \emph{current density} $\mathbf{J}$ it is natural to work with the $4$-vector $J=(\rho, \vec{J})$ called \emph{charge-current density}. Also, the electric field $\mathbf{E}$ and the magnetic field $\mathbf{B}$ are combined into a $2$-form $F$ called \emph{field strength tensor} defined as:
\begin{equation}
    F=\diff t \wedge (E_x\diff x+E_y\diff y+E_z\diff z)+B_x\diff y \wedge \diff z+B_y\diff z \wedge \diff x+B_z\diff x \wedge \diff y.
\end{equation}
In this framework, the four Maxwell equations are equivalent to the Yang-Mills equations \eqref{eq:YangMillsEqs}.

An electromagnetic potential $A \in \Omega^1(M)$ is an appropriate potential for the field strength tensor $F$ if $\diff A = F$. It is not unique, since for any smooth function $\phi \in \mathcal{C}^\infty(M)$ the form $A + \diff \phi$ will be again an appropriate potential for $F$. To describe this gauge, it is convenient to define a connection on a principal $G$-bundle $\pi\colon P \longrightarrow \mathbf{R}^4$.

Take $G=U(1)$, so $\mathfrak{g} = \mathfrak{u}(1) \simeq \mathbf{R}$, and chose a connection on $P$ with connection form $\theta$. The Minkowski space $\mathbf{R}^4$ is contractible, hence, $\pi\colon P \longrightarrow \mathbf{R}^4$ is globally trivializable. For any trivialization corresponding to a section $s\colon \mathbf{R}^4 \longrightarrow P$, the $1$-form $s^* \theta$ is an $\mathbf{R}$-valued form on $\mathbf{R}^4$, so it can be written as $s^* \theta = A$. Then, any other section $s'$ can be written as $x \longmapsto s(x) \cdot u(x)^{-1}$ for some function $u\colon x \longmapsto \mathrm{e}^{\mathrm{i} \phi(x)} \colon \mathbf{R}^4 \longmapsto U(1)$ and we have
\begin{equation}
    s'^*(\theta)=A+\diff (\mathrm{e}^{\mathrm{i}\phi})\cdot \mathrm{e}^{-\mathrm{i}\phi}= A + \mathrm{i} \diff \phi.
\end{equation}

The choice of an electromagnetic potential is then equivalent to a choice of a trivialization of the principal $U(1)$-bundle $\pi\colon P \longrightarrow \mathbf{R}^4$. Identifying $P$ with $\mathbf{R}^4\times G$ via a section $s$, we can rewrite the field strength tensor as $F=s^*\omega=s^*\diff \theta=\diff A$.

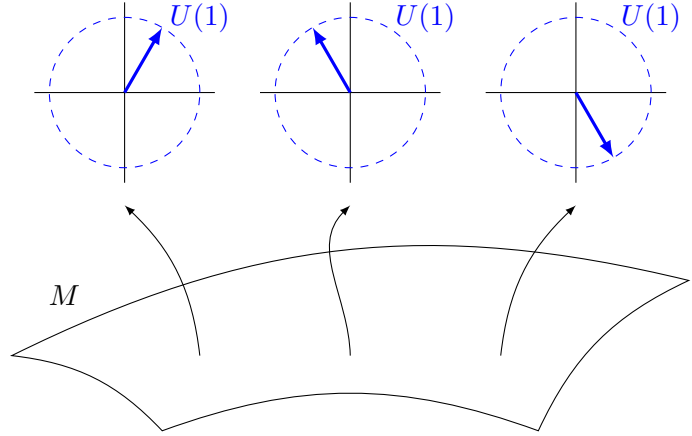
\begin{figure}[ht!]
\centering
\begin{tikzpicture}
  \begin{scope}[shift={(-0.5, 0)}]
    \draw (-4, 2) to[relative, out=20, in=160]
          (-2, 1) to[relative, out=20, in=160]
          ( 3, 1) to[relative, out=20, in=160] 
          ( 5, 3) to[relative, out=-20, in=-160] 
          cycle;
    \draw (-4, 2) node[xshift=0.7cm,yshift=0.8cm]{$M$};
  \end{scope}

  \draw[->] (-2, 2) to[relative, out=-20, in=-160] node[left] {} (-3, 4);
  \draw[->] (0, 2) to[relative, out=5 in=190] node[above] {} (0, 4);
  \draw[->] (2, 2) to[relative, out=20, in=160] node[right] {} (3, 4);

  \begin{scope}[shift={(0, +2.5)}]
    \draw[-] (-3,1.8) -- (-3,4.2);
    \draw[-] (-4.2,3) -- (-1.8,3);
    \draw[<-,color=blue, very thick] (60:1)++(-3,3) --(-3,3);
    \draw[color=blue, dashed] (-3,3) circle (1 cm) node[xshift=1cm,yshift=1cm] {$U(1)$};
    
    \draw[-] (0,1.8) -- (0,4.2);
    \draw[-] (-1.2,3) -- (1.2,3);
    \draw[<-,color=blue, very thick] (120:1)++(0,3) --(0,3);
    \draw[color=blue, dashed] (0,3) circle (1 cm) node[xshift=1cm,yshift=1cm] {$U(1)$};

    \draw[-] (3,1.8) -- (3,4.2);
    \draw[-] (1.8,3) -- (4.2,3);
    \draw[<-,color=blue, very thick] (300:1)++(3,3) --(3,3);
    \draw[color=blue, dashed] (3,3) circle (1 cm) node[xshift=1cm,yshift=1cm] {$U(1)$};
  \end{scope}
\end{tikzpicture}  
\caption{In blue, the \emph{phasor} representing the amplitude and phase of an electromagnetic signal. The principal $U(1)$-bundle is the natural phase space for electrodynamics. Sections of this bundle are often called phasors in the jargon of physics and engineering.}
\label{fig:U1-bundle}
\end{figure}
\end{example}


\begin{lemma}
    Let $\pi\colon P \longrightarrow M$ be a principal $G$-bundle over an $E$-manifold $(M, E_M)$ with the pullback structure $E_P$.
    \begin{enumerate}
        \item The action $r\colon P \times G \longrightarrow P$ is an $E$-action.
        \item The vector bundle map
        \begin{equation}
            \begin{array}{rccc}
                \varphi\colon & P \times \mathfrak{g} & \longrightarrow & \ker \diff \pi \\
                 & (p, \alpha) & \longmapsto & \diff_e \psi_p (\alpha)
            \end{array}
        \end{equation}
        is an isomorphism. Here, $\psi_p$ is the orbit map at a point $p \in P$. \qedhere
    \end{enumerate}
\end{lemma}

\begin{proof}
    Consider the $E$-structure $\operatorname{Vec}(G) \times E_P$ in $G \times P$. The proof of the first item follows trivially from the $G$-invariance condition. 
    
    Regarding the second item, we now have that, if $\tilde{\tau}\colon \mathrm{T}P \longrightarrow \mathrm{T}M$ is the standard tangent map, then $\ker \tilde{\tau} \subset E_P$. The proof now follows from the smooth case.
\end{proof}

A fundamental result in the symplectic formulation of classical gauge theories is the existence of horizontal lifts and connection forms. Their construction and properties can be fundamentally derived from a result in homological algebra known as the splitting lemma. We show in analogy with the smooth setting that, even though the category of vector bundles is not abelian because of the possible non-existence of kernels, a splitting lemma still holds for $E$-Atiyah sequences.

\begin{proposition}[$E$-Splitting lemma] \label{prop:SplitLemma}
    Consider a principal $G$-bundle $\pi\colon P \longrightarrow M$ over an $E$-manifold $(M, E_M)$ with the pullback structure $E_P$. The following objects are in one-to-one correspondence:
    \begin{enumerate}
        \item $G$-Invariant splittings of the short exact sequence \eqref{eq:EAtiyahSeq}.
        \item $G$-Invariant retractions $\theta\colon \ee{T}P \longrightarrow P \times \mathfrak{g}$.
        \item $G$-Invariant sections $h\colon \pi^* \ee{T}M \longrightarrow \ee{T}P$. \qedhere
    \end{enumerate}
\end{proposition}

\begin{notation}
    Consider a principal connection \eqref{eq:EAtiyahSeq} and the induced connection in $\ee{T}P$. The induced retraction $\theta\colon \ee{T}P \longrightarrow \ker \diff \pi$ is called the \emph{connection form}. It is an element of $\Omega^1(P \times \mathfrak{g}; \ee{T}P)$. The section $h$ is called the \emph{horizontal lift} of the connection.
\end{notation}

\begin{remark} \label{rmk:SplitLemma}
    We will consider from now and onward the natural left action induced from $r\colon P \times G \longrightarrow P$, defined as $l_g(p) = r_{g^{-1}}(p)$ for any $p \in P$ and $g \in G$. The splitting $\Phi$ of the $E$-Atiyah sequence and the section $h$ remain invariant by the action of $l$, while the connection form $\theta$ now intertwines $l$ and the adjoint action in $P \times \mathfrak{g}$.
\end{remark}

\subsection{The Hamiltonian formalism of gauge theories} \label{ssec:WeinsteinFormalism}



Having proved the main results concerning the structure of $E$-principal $G$-bundles and gauge fields over them, we are ready to investigate the symplectic formulation of the equations of motion. Throughout the subsection we will consider a fixed principal $G$-bundle $\pi\colon P \longrightarrow M$ with the pullback structure $E_P$. The following commutative diagram contains the information of the splitting lemma \ref{prop:SplitLemma},

\begin{equation} \label{eq:AtiyahSplit}
    \begin{tikzcd}
    	&& {\ee{T}P} \\
    	0 & {P \times \mathfrak{g}} && {\pi^* \ee{T}M} & 0 \\
    	&& {(P \times \mathfrak{g})\oplus \mathrm{H}P}
    	\arrow[from=2-1, to=2-2]
    	\arrow["{i \varphi}"', shift right=1, from=2-2, to=1-3]
    	\arrow["{i_1}", shift left=1, from=2-2, to=3-3]
    	\arrow["{\pi_2}"', shift right=1, from=3-3, to=2-4]
    	\arrow["{\tilde{\pi}}", shift left=1, from=1-3, to=2-4]
    	\arrow[from=2-4, to=2-5]
    	\arrow["\Phi", from=1-3, to=3-3]
    	\arrow["{\pi_1}", shift left=1, from=3-3, to=2-2]
    	\arrow["{i_2}"', shift right=1, from=2-4, to=3-3]
    	\arrow["h", shift left=1, from=2-4, to=1-3]
    	\arrow["\theta"', shift right=1, from=1-3, to=2-2]
    \end{tikzcd}
\end{equation}

where the maps are understood to be $G$-invariant. In this diagram, the splitting gives an identification $\mathrm{H}P \simeq \pi^* \ee{T}M$. As a consequence, $\mathrm{H}P$ is an $E$-manifold and the splitting $\Phi$ is an $E$-map. The isomorphism $\Phi$ also induces a $G$-action in $(P \times \mathfrak{g}) \oplus \mathrm{H}P$.

In classical Hamiltonian mechanics, the phase space for Hamilton's equations of motion over a manifold $M$ is the cotangent bundle $\mathrm{T}^*M$. In the extension of Hamiltonian mechanics to $E$-manifolds, we have seen through theorem \ref{thm:EGeod} that the natural phase space becomes $\ee{T}^*M$. We have also concluded that the configuration space of a particle in a gauge theory has to be enlarged to the total bundle $\tau\colon B \longrightarrow M$. In this setting, the natural phase space to describe Hamiltonian dynamics is $\ee{T}B$. Consequently, we are interested in the dual splitting defined by a gauge field, represented in the following commutative diagram:
\begin{equation} \label{eq:DualAtiyahSplit}
    \begin{tikzcd}
    	&& {\ee{T}^*P} \\
    	0 & {\pi^* \ee{T}^*M} && {P \times\mathfrak{g}^*} & 0 \\
    	&& {(P \times \mathfrak{g}^*) \oplus P^\sharp}
    	\arrow[from=2-1, to=2-2]
    	\arrow[from=2-4, to=2-5]
    	\arrow["\Psi"', from=3-3, to=1-3]
    	\arrow["{\pi_2^\dagger}", shift left=1, from=2-2, to=3-3]
    	\arrow["{i_1^\dagger}"', shift right=1, from=3-3, to=2-4]
    	\arrow["{\tilde{\pi}^\dagger}"', shift right=1, from=2-2, to=1-3]
    	\arrow["{\varphi^\dagger i^\dagger}", shift left=1, from=1-3, to=2-4]
    	\arrow["{\pi_1^\dagger}"', shift right=1, from=2-4, to=3-3]
    	\arrow["{\theta^\dagger}", shift left=1, from=2-4, to=1-3]
    	\arrow["{h^\dagger}"', shift right=1, from=1-3, to=2-2]
    	\arrow["{i_2^\dagger}", shift left=1, from=3-3, to=2-2]
    \end{tikzcd}
\end{equation}
We have denoted $P^\sharp = \mathrm{H}^* P$. Moreover, it can be noted that the induced splitting is given by $\Psi = \Phi^\dagger$. As before, the isomorphism $\Psi$ induces an action of $G$ in $(P \times \mathfrak{g}^*)\oplus P^\sharp$ and a symplectic form by pushforward.

Before presenting the main result, we will show that the vector bundle $P^\sharp$ is a specific realization of the pullback bundle of $\ee{T}^*M$ by the projection $\pi\colon P \longrightarrow M$.

\begin{lemma}
    Let $\pi\colon P \longrightarrow M$ be an $E$-principal $G$-bundle over an $E$-manifold $(M, E_M)$. In the notation of diagrams \eqref{eq:AtiyahSplit} and \eqref{eq:DualAtiyahSplit}, the vector bundle $P^\sharp$ fits the pullback bundle diagram
    \begin{equation} \label{eq:PSharpPllb}
        \begin{tikzcd}[sep=large]
        	{P^\sharp} & P \\
        	{\ee{T}^* M} & M
        	\arrow["{\pi^\sharp}"', from=1-1, to=2-1]
        	\arrow["{\tau_{\ee{T}^*M}}"', from=2-1, to=2-2]
        	\arrow["{\pi}", from=1-2, to=2-2]
        	\arrow["{\tau_{\ee{T}^*P}\tilde{\pi}^\dagger}", from=1-1, to=1-2]
        \end{tikzcd}
    \end{equation}
    Here, the map $\pi^\sharp\colon P^\sharp \longrightarrow \ee{T}^*M$ is defined by its action on elements $X \in \mathfrak{g} \times \mathrm{H}P$ and $\beta \in \mathfrak{g}^* \times P^\sharp$ as
    \begin{equation} \label{eq:pllbDef}
        \langle \beta, X \rangle = \langle \kappa \tilde{\pi}^\dagger \beta, \diff \pi h X \rangle. \qedhere
    \end{equation}
\end{lemma}

\begin{proof}
    To begin the proof we realize that, by definition, $\pi^* \ee{T}^*M$ is the pullback bundle of $\ee{T}^*M$ by the submersion $\pi\colon P \longrightarrow M$. We take the identification $P^\sharp \simeq \pi^*\ee{T}^*M$ and construct the following commutative diagram, where we have also expressed the pullback property of $\pi^*\ee{T}M$:
    \begin{equation} \label{eq:PllbDiagrams}
        \begin{tikzcd}[sep=large]
        	{\ee{T}P} &&& {\ee{T}^*P} \\
        	{} & {\pi^* \ee{T} M} & P && {P^\sharp} & P \\
        	& {\ee{T}M} & M && {\ee{T}^*M} & M
        	\arrow[from=2-2, to=2-3]
        	\arrow["\pi", from=2-3, to=3-3]
        	\arrow["\sigma"', from=2-2, to=3-2]
        	\arrow[from=3-2, to=3-3]
        	\arrow["{\tilde{\pi}}"', shift right=1, from=1-1, to=2-2]
        	\arrow["{\diff \pi}"', curve={height=18pt}, from=1-1, to=3-2]
        	\arrow["h"', shift right=1, from=2-2, to=1-1]
        	\arrow["{\pi^\sharp}"', from=2-5, to=3-5]
        	\arrow["{\tau_{P^\sharp}}", from=2-5, to=2-6]
        	\arrow[from=3-5, to=3-6]
        	\arrow["\pi", from=2-6, to=3-6]
        	\arrow["{h^\dagger}"', shift right=1, from=1-4, to=2-5]
        	\arrow["{\tilde{\pi}^\dagger}"', shift right=1, from=2-5, to=1-4]
        	\arrow["{\tau_{\ee{T}^*P}}", curve={height=-18pt}, from=1-4, to=2-6]
        	\arrow["\kappa"', curve={height=18pt}, from=1-4, to=3-5]
        \end{tikzcd}
    \end{equation}
    
    The proof of the commutativity of these diagrams now follows from the same elementary idea: the maps in the splittings \eqref{eq:AtiyahSplit} and \eqref{eq:DualAtiyahSplit} are vector bundle morphisms. Firstly, as $\pi^\dagger$ is a vector bundle morphism it is straightforward to see that $\tau_{P^\sharp} = \tau_{\ee{T}^*P} \tilde{\pi}^\dagger$. We are defining the projection $\kappa\colon \ee{T}^*P \longrightarrow \ee{T}^*M$ by means of the dual splitting $h^\dagger$.
    
    To prove equation \eqref{eq:pllbDef} we use the fact that $\ee{T}M$ and $P^\sharp$ are pullback bundles and, as a consequence, $\langle \beta, X \rangle = \langle \pi^\sharp \beta, \sigma X \rangle$. Following diagrams \eqref{eq:PllbDiagrams} we see that $\sigma = \diff \pi h$, which proves the result. 
\end{proof}

\begin{remark} \label{rmk:ExtensionPllb}
    To prove equation \eqref{eq:pllbDef} we have used the inclusions $h, \tilde{\pi}^\dagger$ to obtain representatives of $X$ and $\beta$ in $\ee{T}P$ and $\ee{T}^*P$, respectively. We can show that this choice of representatives is not relevant for expression \eqref{eq:pllbDef}, that is, for any $Y \in \ee{T}P$ and $\eta \in \ee{T}^*P$ such that $\tilde{\pi} Y = X$ and $h^\dagger \eta = \beta$ we have $\langle \beta, X \rangle = \langle \kappa \eta, \diff \pi Y \rangle$.
    
    To prove this fact we only have to see, given equation \eqref{eq:pllbDef}, that $\langle \kappa \eta, \diff \pi Y \rangle = \langle \kappa \tilde{\pi}^\dagger \beta, \diff \pi h X \rangle$. A direct computation shows
    \begin{align*}
        \langle \kappa \eta, \diff \pi Y \rangle - \langle \kappa \tilde{\pi}^\dagger \beta, \diff \pi h X \rangle &= \langle \kappa \eta, \diff \pi Y \rangle - \langle \kappa \tilde{\pi}^\dagger \beta, \diff \pi Y \rangle + \langle \kappa \tilde{\pi}^\dagger \eta, \diff \pi Y \rangle - \langle \kappa \tilde{\pi}^\dagger \beta, \diff \pi h X \rangle \rangle \\
        &= \langle \kappa (\eta - \tilde{\pi}^\dagger \beta), \diff \pi Y \rangle - \langle \kappa \tilde{\pi}^\dagger \beta, \diff \pi(Y -  h X) \rangle.
    \end{align*}
    We use now the fact that $\tilde{\pi} h = \id_{\pi^* \ee{T}M}$ from the properties of the splitting lemma \ref{prop:SplitLemma}. From diagram \eqref{eq:PllbDiagrams} and using $\tilde{\pi} Y = X$, we have $\diff\pi(Y - h X) = \sigma \tilde{\pi} Y - \sigma \tilde{\pi} h X = \sigma X - \sigma X = 0$. Similarly, $h^\dagger \tilde{\pi}^\dagger = \id_{P^\sharp}$. Using condition $h^\dagger \eta = \beta$ we can show that $\kappa(\eta - \tilde{\pi}^\dagger \beta) = \pi^\sharp h^\dagger \eta - \pi^\sharp h^\dagger \tilde{\pi}^\dagger \beta = \pi^\sharp \beta - \pi^\sharp \beta = 0$. This concludes the proof.
\end{remark}

\begin{theorem} \label{thm:EWeinsteinUniversal}
    Consider a principal $G$-bundle $\pi\colon P \longrightarrow M$ over an $E$-manifold $M$ and a Hamiltonian $G$-space $Q$.
    \begin{enumerate}
        \item The product space $\ee{T}^*P \times Q$ is Hamiltonian with moment map $\mu_P + \mu_Q$.
        \item The hypothesis of the reduction theorem \ref{thm:ERed} are satisfied and, consequently, the space $(\ee{T}^*P \times Q)_0$ is an $E$-symplectic manifold. 
        \item The horizontal lift $h^\dagger$ is well-defined in classes of equivalence and defines a map $\alpha\colon (\ee{T}^*P \times Q)_0 \longrightarrow \ee{T}^*M$. \qedhere
    \end{enumerate}
\end{theorem}

\begin{proof}
    The first item is straightforward from lemma \ref{lem:ProdHamAct}. We simply take the natural $E$-manifold structure $(Q, \operatorname{Vec}(Q))$.
    
    To see that we can define the Marsden-Weinstein reduced space we have to check that we are in the conditions of theorem \ref{thm:ERed}. Given that $0 \in \mathfrak{g}^*$ is a regular value of $\mu_P$, it is also a regular value of $\mu_P + \mu_Q$. As $G$ acts transitively in $\ee{T}^*P$, it also acts transitively in $\ee{T}^*P \times Q$. This shows that the reduced space $\mu^{-1}(0)/G$ is well-defined and carries a natural $E$-symplectic structure induced from that of $\ee{T}^*P \times Q$.
    
    Regarding the surjection $\alpha\colon (\ee{T}^*P \times Q)_0 \longrightarrow \ee{T}^*M$, we consider the composition of the surjection $h^\dagger\colon \ee{T}^*P \longrightarrow P^\sharp$ induced from the connection and the natural projection $\tilde{p}_1\colon \ee{T}^*P \times Q \longrightarrow \ee{T}^*P$, giving the map $h^\dagger \tilde{p}_1 \colon \ee{T}^*P \times Q \longrightarrow P^\sharp$. We can consider the restriction of this map to the submanifold $j\colon \mu^{-1}(0) \longrightarrow \ee{T}^*P \times Q$, giving the map $h^\dagger \tilde{p}_1 j$. As $h^\dagger$, $\tilde{p}_1$, and $j$ intertwine the $G$-actions in their respective manifolds, we can take the map defined in classes $\alpha = [h^\dagger \tilde{p}_1 j]$, which completes the proof.
\end{proof}

\begin{remark}
    Notice that in the hypotheses of theorem \ref{thm:EWeinsteinUniversal} we could have taken $Q$ to be an $E$-Hamiltonian $G$-space with $E$-structure $E_Q$ compatible with the $G$-action. This choice implies that the associated bundle $B = P \times_G Q$ is no longer equipped with the pullback $E$-structure \ref{def:ProlFibreBun}, but rather with a compatible $E$-structure $\tilde{E}_B \subset E_B$. This procedure induces an $E$-structure in the typical fibre $Q$, which yields room for more levels of degeneracy in the construction.
\end{remark}


Weinstein proved the equivalence of this construction with the so-called Sternberg space, originally introduced in \cite{SternbergMinCoup}. An analogue statement, which we present here, holds in the context of $E$-manifolds.

\begin{theorem} \label{thm:WeinsIsom}
    Let $\pi\colon P \longrightarrow M$ be an $E$-principal $G$-bundle and consider a $G$-Hamiltonian space $Q$. There exists a diffeomorphism
    \begin{equation}
        \mu^{-1}(0) \simeq P^\sharp \times Q.
    \end{equation}
    As a consequence, there exists a symplectomorphism of spaces $(\ee{T}^* P \times Q)_0 \simeq P^\sharp \times_G Q$.
\end{theorem}

\begin{proof} 
    
    Notice that, from diagram \eqref{eq:DualAtiyahSplit}, that the horizontal lift induces a map $h^{\dagger}\colon \ee{T}^*P \longrightarrow P^\sharp$. In order to obtain the isomorphism with Sternberg's space we notice that, because we have chosen the trivialization $\ker \diff \pi \simeq P \times \mathfrak{g}^*$ by the infinitesimal action of $G$ on $\ee{T}P$, the moment map of the action is simply $\mu_P = - p_2$ (see \cite{WeinsteinUniversal}). As a consequence, the moment map $\mu\colon \mathfrak{g}^* \times P^\sharp \times Q \longrightarrow \mathfrak{g}^*$ is given by $\mu = -p_1 + \mu_Q$. With this expression, we can readily show that $\mu^{-1}(0) \simeq P^\sharp \times Q$. Of course, if $\mu(p, X, q) = -X + \mu_Q(q) = 0$ we have that $p \in P^\sharp$ is arbitrary and that $X = \mu_Q(q)$. As a consequence, $q \in Q$ is arbitrary and $X \in \mathfrak{g}^*$ is completely determined by $q$, so the claim is proved.
\end{proof}


The results proved here can be framed under the paradigm of ``cotangent bundle reduction'', as referred to in \cite{MarsdenCotangent}. The authors study the reduction in the cotangent bundle of a smooth manifold $M$ by a proper and free Lie group action $G$. As the quotient $\pi\colon M \longrightarrow M/G$ defines a principal bundle, this model is completely equivalent to our setting. The reduction of the cotangent action by a coadjoint orbit is can be re-conducted to the reduction at $0$ using the shifting trick presented in Theorem \ref{thm:EShiftTrick}. Specifically, the diffeomorphism $\mu^{-1}(\alpha)/G_\alpha \simeq \mu^{-1}(\mathcal{O}(\alpha))/G$ used in the proof together with the shifting trick is the main tool connecting our results and those found in \cite{MarsdenCotangent}. We briefly outline some generalizations in this direction to the theory of $E$-manifolds.

\begin{corollary}
    Let $(M, E_M)$ be an $E$-manifold and consider a principal $G$-bundle $\pi\colon P \longrightarrow M$ with the induced structure $E_P$. Then, the map
    \begin{equation}
        \begin{array}{rccc}
            \sigma\colon & P \times \mathfrak{g}^* & \longrightarrow & \ee{T}^* P/G \\
             & (p, \alpha) & \longmapsto & [\theta^\dagger \alpha]
        \end{array}
    \end{equation}
    intertwines the diagonal action and the cotangent lifted action of $G$ on $\ee{T}^* P$, and therefore it factors to a quotient map $[\sigma]\colon P \times_G \mathfrak{g}^* \longrightarrow \ee{T}^* P/G$ which is fibre-wise injective and $[\varphi^\dagger i^\dagger] [\sigma] = \id$.
\end{corollary}

\begin{proof}
    As it was noted in remark \ref{rmk:SplitLemma}, the connection form $\theta$ associated to a connection intertwines the tangent lift in $\ee{T}P$ and the diagonal action in $P \times \mathfrak{g}$. A direct consequence is that $\theta^\dagger$ intertwines the diagonal action in $P \times \mathfrak{g}$ and the cotangent lift in $\ee{T}^* P$. Moreover, this map is a splitting of the dual Atiyah sequence \eqref{eq:DualAtiyahSplit} and, as a consequence, is fibre-wise injective and $\varphi^\dagger i^\dagger \theta^\dagger = \id$. The result follows now taking quotients and noticing that $[\sigma] = [\theta^\dagger]$.
\end{proof}

\begin{corollary}
    If $\mathcal{O} \subset \mathfrak{g}^*$ is a coadjoint orbit, then $\mu^{-1}(\mathcal{O})/G \simeq \ee{T}^*M \times_M (P \times_G \mathcal{O})$.
\end{corollary}

\begin{proof}
    The proof follows the same guidelines as the one of theorem \ref{thm:WeinsIsom} and, in fact, the result to be proved now is a consequence of \ref{thm:WeinsIsom} using the shifting trick (Theorem \ref{eq:ShiftTrick}). We will, however, give a more explicit description. As in the proof of \ref{thm:WeinsIsom}, the moment map of the cotangent lift in the trivialization $\ee{T}^*P \simeq P^\sharp \oplus (P \times \mathfrak{g}^*)$ is given by $\mu\colon (X, p, \alpha) \longmapsto \alpha$. As a consequence, we can readily see that $\mu^{-1}(\mathcal{O}) \simeq P^\sharp \oplus (P \times \mathcal{O})$. The result follows taking quotients.
\end{proof}

\subsection{The Montgomery space for Yang-Mills theories}

Classical Yang-Mills theories fall within the setting of gauge theories. In a Yang-Mills theory, the fundamental object is a principal $G$-bundle $\pi\colon P \longrightarrow M$ over a Riemannian manifold $(M, g)$ and the configuration space is the associated bundle $P \times_G \mathfrak{g}^*$~\cite{GSSymplecticTechniques}. As a consequence, the internal symmetries or charges of a particle are described by elements of the dual algebra $\mathfrak{g}^*$. In such theories, the gauge field $\nabla$ is a connection specified following the Yang-Mills equations,
\begin{equation} \label{eq:YangMillsEqs}
    \left\{ \begin{array}{r@{}c@{}l}
        \diff_\nabla F & {}={} & 0, \\
        \star \diff_\nabla \star F & {}={} & J,
    \end{array} \right.
\end{equation}
where $F$ is the curvature of the connection $\nabla$ and $J$ is a source term for the equations. The operator $\star\colon \Omega^\bullet(M) \longrightarrow \Omega^{\bullet-1}(M)$ is Hodge's star operator induced from the metric $g$.

Montgomery exhibited the natural phase space of a classical particle under a Yang-Mills field in~\cite{MontgomeryWongsEquations} as a symplectic leaf of a bigger Poisson space and generalized Weinstein's isomorphism (theorem \ref{thm:WeinsIsom}) to Poisson manifolds. The following theorem is an extension of these results to the setting of Yang-Mills theories on $E$-manifolds.

\begin{theorem} \label{thm:MinCoupMont}
    Consider an $E$-principal $G$-bundle $\pi\colon P \longrightarrow M$. An $E$-principal connection induces the commutative diagram \ref{eq:PoissMinCoup}. Moreover, the map $[\Psi]$, called the \emph{minimal coupling}, is a Poisson isomorphism.
    \begin{equation} \label{eq:PoissMinCoup}
        \begin{tikzcd}[sep=small]
        	{\ee{T}^*P} && {\ee{T}^*P/G} \\
        	&&& {\ee{T}^*M} \\
        	{\mathfrak{g}^* \times P^\sharp} && {\mathfrak{g}^* \times_G P^\sharp}
        	\arrow["{\pi_{\ee{T}^*P/G}}", from=1-1, to=1-3]
        	\arrow["{\pi_{\mathfrak{g}^* \times_G P^\sharp}}"', from=3-1, to=3-3]
        	\arrow["{[\pi_2]}"', from=3-3, to=2-4]
        	\arrow["{[h^\dagger]}", from=1-3, to=2-4]
        	\arrow["{[\Psi]}", from=3-3, to=1-3]
        	\arrow["\Psi", from=3-1, to=1-1]
        \end{tikzcd}
    \end{equation}
    As a result,
    \begin{itemize}
        \item The Hamiltonian equations of motion generated by the pullback of a Hamiltonian function $H \in \mathcal{C}^\infty(\ee{T}^*M)$ to $\ee{T}^*P/G$ or to $\mathfrak{g}^* \times_G P^\sharp$ are equivalent. These equations are called \emph{Wong's equations of motion}.
        \item We obtain induced dynamics in the associated bundle $P \times_G \mathfrak{g}$, which has base $M$ and fibre $\mathfrak{g}$. Elements in the fibres are called \emph{charges}\footnote{In the case $G = \operatorname{U}(1)$ the fibre is $\mathfrak{u}^*(1) \simeq \mathbf{R}$, which is the standard electric charge. For the strong nuclear force, $G = \operatorname{SU}(3)$ and elements in $\mathfrak{su}^*(3)$ are called \emph{color charges}.}. \qedhere
    \end{itemize}
\end{theorem}


\begin{proof}
    Consider the induced Poisson structure $\Pi_{\ee{T}^*P} \in \ee{\operatorname{Vec}}^2(\ee{T}^* P)$ from the canonical symplectic form $\omega_P$. This structure induces a Poisson structure by pull-back $\Pi_{\mathfrak{g}^* \times P^\sharp} = \diff \Psi^{-1} \Pi_{\ee{T}^*P} \in \ee{\operatorname{Vec}}^2(\mathfrak{g}^*\times P^\sharp)$. As the action of $G$ in $\ee{T}^* P$ is symplectic (it is a cotangent lift), the quotient $\ee{T}^*P/G$ can be endowed a Poisson structure $\Pi_{\ee{T}^*P/G} = \diff \pi_{\ee{T}^*P/G} \Pi_{\ee{T}^*P}$. The induced $G$-action in $\mathfrak{g}^* \times P^\sharp$ induces a surjection $\pi_{\mathfrak{g}^* \times_G P^\sharp}$ which can be used to define a Poisson structure $\Pi_{\mathfrak{g}^* \times_G P^\sharp} = \diff \pi_{\mathfrak{g}^* \times_G P^\sharp} \Pi_{\mathfrak{g}^* \times P^\sharp}$. The fact that this structure exists also follows from the fact that the $G$-action in $\mathfrak{g}^* \times P^\sharp$ is symplectic. The map $\Psi$ intertwines the orbits of $G$ in $\ee{T}^*P$ and $\mathfrak{g}^* \times P^\sharp$ and, as a consequence, there exists a map $[\Psi]\colon \mathfrak{g}^* \times_G P^\sharp \longrightarrow \ee{T}^* P/G$ which is well-defined on orbits and makes diagram \eqref{eq:PoissMinCoup} commute. Given that the projections to the quotients $\pi_{\ee{T}^*P/G}$, $\pi_{\mathfrak{g}^* \times_G P^\sharp}$, and $\Psi$ are Poisson maps, the induced map $[\Psi]$ is a Poisson morphism as
    \begin{equation}
        \diff [\Psi] (\Pi_{\mathfrak{g}^* \times_G P^\sharp}) = \diff [\Psi] \diff \pi_{\mathfrak{g}^* \times_G P^\sharp} (\Pi_{\mathfrak{g}^* \times P^\sharp}) = \diff \pi_{\ee{T}^* P/G} \diff \Psi (\Pi_{\mathfrak{g}^* \times P^\sharp}) = \Pi_{\ee{T}^*P/G}. \qedhere
    \end{equation}
\end{proof}

The equivalence given by the minimal coupling procedure gives two different ways to understand the action of a Yang-Mills field. Under Weinstein's perspective, the field defines a way to pull back a Hamiltonian function to $\ee{T}^*P$. In this sense, the field ``modifies'' the Hamiltonian function but leaves the Poisson structure unchanged. On the other hand, from Sternberg's viewpoint the Hamiltonian function is ``unchanged'' by the pullback to the space $\mathfrak{g}^* \times P^\sharp$. Rather, the Yang-Mills field modifies the Poisson structure by a magnetic term. The following result describes the induced Poisson structure and is a direct generalization of the expression given in \cite{MontgomeryThesis}.

\begin{proposition}
    Let $\pi\colon P \longrightarrow M$ be an $E$-principal $G$-bundle. Denote by $\lambda_M, \lambda_P$ the canonical Liouville forms in $\ee{T}^*M$ and $\ee{T}^*P$, respectively, and by $\omega_M$ and $\omega_P$ their canonical symplectic forms.
    \begin{enumerate}
        \item We have the equality $(\theta^\dagger)^*\lambda_P = \langle \alpha, \theta^\sharp \rangle$. Here, $\theta^\sharp$ is the pullback connection $\theta^\sharp = \tau_{P \times \mathfrak{g}^*}^* \theta$ and $\langle \alpha, \theta^\sharp \rangle$ is defined as the one-form $\langle \alpha, \theta^\sharp \rangle \in \ee{\Omega}^1(\mathfrak{g}^* \times P^\sharp)$ whose action on elements is $(\alpha, X) \longmapsto \langle\alpha, \theta^\sharp(X) \rangle$.
        \item We have the equality $(\tilde{\pi}^\dagger)^* \lambda_P = (\pi^\sharp)^* \lambda_M$.
        \item The induced Liouville form in $(P \times \mathfrak{g}^*) \oplus P^\sharp$ by the isomorphism $\Psi$ is
        \begin{equation} \label{eq:IndLiouForm}
            \Psi^* \lambda_P = (\pi^\sharp)^* \lambda_M + \langle \alpha, \theta^\sharp \rangle.
        \end{equation}
        As a consequence, the induced symplectic form is expressed as
        \begin{equation} \label{eq:IndSympForm}
            \Psi^* \omega_P = (\pi^\sharp)^* \omega_M + \diff \langle \alpha, \theta^\sharp \rangle. \qedhere
        \end{equation}
    \end{enumerate}
\end{proposition}

\begin{proof}
    To prove the first item we show by a direct computation, taking $(p, \alpha) \in P \times \mathfrak{g}^*$ and $X \in \mathrm{T}_\alpha(P \times \mathfrak{g}^*)$, that
    \begin{equation*}
        (\theta^\dagger)^* (\lambda_P)_{(p, \alpha)} (X) = \big\langle (\lambda_P)_{\theta^\sharp(p, \alpha)}, \diff \theta^\dagger X \big\rangle = \langle \theta^\dagger (p, \alpha), \diff \tau_{\ee{T}^*P} \diff \theta^\dagger X \rangle = \langle \alpha, \theta \diff \tau_{\ee{T}^*P} \diff \theta^\dagger X \rangle
    \end{equation*}
    Now, we can use that $\theta^\dagger$ is a vector bundle morphism to get that $\tau_{\ee{T}^*P} \theta^\dagger = \tau_{P \times \mathfrak{g}^*}$. As a consequence, $\theta \diff \tau_{\ee{T}^*P} \diff \theta^\dagger = \tau_{P \times \mathfrak{g}^*}^* \theta$, concluding the proof.
    
    Regarding the second item, we compute both terms separately. Let us choose a point $\beta \in \ee{T}^*_p P$ and a tangent vector $X \in \ee{T}_\beta \ee{T}^*P$. A direct computation using the definition of Liouville form and the conclusion of remark \ref{rmk:ExtensionPllb} yields
    \begin{equation*}
        (\tilde{\pi}^\dagger)^* (\lambda_P)_\beta (X) = \big\langle (\lambda_P)_{\tilde{\pi}^\dagger \beta}, \diff \tilde{\pi}^\dagger X \big\rangle = \langle \tilde{\pi}^\dagger \beta, \diff \tau_{\ee{T}P} \diff \tilde{\pi}^\dagger X \rangle = \langle \pi^\sharp \beta, \diff \pi \diff \tau_{\ee{T}P} \diff \pi^\dagger X \rangle.
    \end{equation*}
    An analogous computation yields
    \begin{equation*}
        (\pi^\sharp)^* (\lambda_M)_\beta (X) = \big\langle (\lambda_M)_{\pi^\sharp \beta}, \diff \pi^\sharp X \big\rangle = \langle \pi^\sharp \beta, \diff \tau_M \diff \pi^\sharp X \rangle.
    \end{equation*}
    The commutativity of diagram \eqref{eq:PSharpPllb} shows that both expresssions agree.
    
    Finally, to prove the third item observe that from the splitting $(P \times \mathfrak{g}^*) \oplus P^\sharp$, any one-form $\alpha$ can be decomposed as $\alpha = (\pi_2^\dagger)^* \alpha + (\pi_1^\dagger)^* \alpha$. Consequently, $\Psi^* \lambda_P = (\Psi \pi_2^\dagger)^* \lambda_P + (\Psi \pi_1^\dagger)^* \lambda_P$ and, as $\Psi \pi_2^\dagger = \tilde{\pi}^\dagger$ and $\Psi \pi_1^\dagger = \theta^\dagger$ from diagram \eqref{eq:DualAtiyahSplit}, we conclude that $\Psi^* \lambda_P = (\pi^\sharp)^* \lambda_P + (\theta^\dagger)^* \lambda_P$. Equation \eqref{eq:IndLiouForm} follows from the first and second items.
    
    Regarding the expression of the induced symplectic form, we only have to use equation \eqref{eq:IndLiouForm} and the commutativity of the pullback of an $E$-map with the exterior differential, as was proved in remark \ref{rmk:ExtensionPllb}. As all the maps in diagram \eqref{eq:DualAtiyahSplit} are $E$-manifolds, we obtain
    \begin{equation*}
        \Psi^* \omega_P = \Psi^* \diff \lambda_P = \diff \Psi^* \lambda_P = \diff (\pi^\sharp)^* \lambda_M + \diff \langle \alpha, \theta^\sharp\rangle = (\pi^\sharp)^* \omega_M + \diff \langle \alpha, \theta^\sharp \rangle. \qedhere
    \end{equation*}
\end{proof}


We show now the minimal coupling procedure in some degenerate instances arising from physical considerations. We also use the induced Poisson structure to compute Wong's equations of motion.

\begin{example}[General minimal coupling for $b$-manifolds] \label{ex:BMinCoupGen}
    Consider a $b$-manifold $(M, Z)$ and a principal $G$-bundle $\pi\colon P \longrightarrow M$. We take now a local chart $(U, \varphi)$ adapted to $Z$ with coordinates $q_1, \ldots, q_n$. We construct natural induced coordinates $\vec{q}, \vec{v}$ in the $b$-tangent bundle $\bi{T}M$. Consider the local trivialization $V = \pi^{-1}(U) \simeq U \times G$ of $P$. Taking the identification $\mathrm{T}G \simeq G \times \mathfrak{g}$ by right-invariant vector fields, we have that $\bi{T}V \simeq \bi{T}U \times G \times \mathfrak{g}$, inducing coordinates $(\vec{q}, \vec{v}, g, \vec{Q})$. The tangent lift $\diff l_\bullet$ is expressed as $\diff l_h (\vec{q}, \vec{v}, g, \vec{Q}) = (\vec{q}, \vec{v}, l_h g, \operatorname{Ad}_g \vec{Q})$. As a consequence, a local expression of the $b$-cotangent bundle is $\bi{T}^* V \simeq \bi{T}^* U \times G \times \mathfrak{g}^*$ and we have local coordinates $(\vec{q}, \vec{p}, g, \vec{O})$. The cotangent lift in these coordinates is $\hat{l}_h(\vec{q}, \vec{p}, g, \vec{O}) = (\vec{q}, \vec{p}, l_h g, \operatorname{Ad}^*_g \vec{O})$. The reduced space is $\bi{T}^*V/G \simeq \bi{T}^*U \times \mathfrak{g}^*$.
    
    The set of coordinates $(U, \varphi)$ also induces a local trivialization $P^\sharp_U \simeq \bi{T}^* U \times G$. Therefore, the adjoint bundle is locally expressed as $P^\sharp_U \times_G \mathfrak{g}^* = \bi{T}^* U \times \mathfrak{g}^*$ and we have coordinates $(\vec{q}, \vec{p}, \vec{Q})$ as used in $\ee{T}V/G$.
    
    
    A connection can be locally specified, by the splitting lemma, as a linear map $h\colon \bi{T} U \longrightarrow \bi{T}V/G$ such that $\diff \pi h = \id_{\bi{T}U}$. In coordinates $(\vec{q}, \vec{v})$ of $\bi{T}U$ and coordinates $(\vec{q}, \vec{v}, \vec{Q})$ of $\bi{T}V/G$ the most general expression for such a map is $h(\vec{q}, \vec{v}) = (\vec{q}, \vec{v}, \vec{Q} + A\cdot \vec{v})$. Here, $A_{i}^j$ are smooth functions of $\vec{q}$ and $\vec{v}$. As a consequence, the splitting of the Atiyah sequence has local expression $[\Phi](\vec{q}, \vec{v}, \vec{Q}) = (\vec{q}, \vec{v}, \vec{Q} + A \cdot \vec{v})$. In particular, $\langle \partial_{v_i}^*, \diff [\Phi] \partial_{v_j} \rangle = \delta_{ij}$, $\langle \partial_{v_i}^*, \diff [\Phi] \partial_{Q_j} \rangle = A_{i}^j$, $\langle \partial_{Q_i}^*, \diff [\Phi] \partial_{v_j} \rangle = 0$, $\langle \partial_{Q_i}^*, \diff [\Phi] \partial_{Q_j} \rangle = \delta_{ij}$. As a consequence, the minimal coupling is given as $[\Psi](\vec{q}, \vec{p}, \vec{O}) = (\vec{q}, \vec{p} + A \cdot \vec{O}, \vec{O})$.
    
    The local expression of the canonical Poisson structure in $\bi{T}^*P/G$ is
    \begin{equation*}
        \Pi_{\ee{T}^*P/G} = q_1 \frac{\partial}{\partial p_1} \wedge \frac{\partial}{\partial q_1} + \sum_{i = 2}^n \frac{\partial}{\partial p_i} \wedge \frac{\partial}{\partial q_i} + \frac{1}{2} \sum_{i,j,k = 1}^n O_k c_{ij}^k \frac{\partial}{\partial O_i} \wedge \frac{\partial}{\partial O_j}.
    \end{equation*}
    From the expression of $\Psi$ in local coordinates and the fact that it is a Poisson map, the induced structure is
    \begin{align*}
        [\Psi]^* \Pi_{\ee{T}^*P/G} &= q_1 \frac{\partial}{\partial p_1} \wedge \frac{\partial}{\partial q_1} + \sum_{i = 2}^n \frac{\partial}{\partial p_i} \wedge \frac{\partial}{\partial q_i} + \frac{1}{2} \sum_{i,j,k = 1}^n O_k F_{ij}^k \frac{\partial}{\partial p_i} \wedge \frac{\partial}{\partial p_j} \\
        & {} \qquad - \frac{1}{2} \sum_{i,j,k,l = 1}^n O_l c_{jk}^l A_i^k \frac{\partial}{\partial p_i} \wedge \frac{\partial}{\partial O_j} + \frac{1}{2} \sum_{i,j,k  =1}^n O_k c^k_{ij} \frac{\partial}{\partial O_i} \wedge \frac{\partial}{\partial O_j}.
    \end{align*}
    This example extends the computations of Montgomery~\cite{MontgomeryThesis} to $b$-manifolds and generalizes the contents of \cite{BraddellKiesenhoferMiranda}.
    
    We can now explicitly compute Wong's equations in local coordinates $(\vec{q}, \vec{p}, \vec{O})$. The pullback of a Hamiltonian function in coordinates $(\vec{q}, \vec{p})$ remains unchanged. Therefore, it has no explicit dependence on $\vec{O}$. A direct computation contracting with the Poisson structure $[\Psi]^* \Pi_{\ee{T}^*P/G}$ shows that
    \begin{align*}
        \dot{q}_1 &= q_1 \frac{\partial H}{\partial p_1}, & \dot{p}_1 &= - q_1 \frac{\partial H}{\partial q_1} + \sum_{j, k, l = 1}^n O_k F_{j1}^k \frac{\partial H}{\partial p_j}, & \dot{O}_i &= \sum_{j, k, l = 1}^n O_l c_{ik}^l A_{j}^k \frac{\partial H}{\partial p_j}, \\
        \dot{q}_i &= \frac{\partial H}{\partial p_i}, & \dot{p}_i &= - \frac{\partial H}{\partial q_i} + \sum_{j, k, l = 1}^n O_k F_{ji}^k \frac{\partial H}{\partial p_j},
    \end{align*}
    with $i = 2, \ldots, n$. We can see the explicit contribution of the boundary in the evolution equations for $q_1$ and $p_1$. Specifically, if $q_1$ = 0, then the point always remains on the boundary of the system, although its conjugated momentum can change due to the interaction with the Yang-Mills field.
\end{example}

\begin{example}[Compactified black hole] \label{ex:penroseblackhole}
Following Sternberg's original motivation, we will compute the minimal coupling of electromagnetism in general relativity. As our singular model, we have chosen  the compactification due to Penrose of a Schwarzschild's solution to a non-rotating black hole. Below we describe its construction, the obtained coordinates and exhibit a new singular structure associated to it (which we already announced in Section 1).

Recall that Schwarzschild's metric can be written in spherical coordinates $(t, r, \theta, \varphi)$ as:
\begin{equation*}
    g = - \bigg( 1 - \frac{2M}{r} \bigg) \diff t^2 + \bigg( 1 - \frac{2M}{r} \bigg)^{-1} \diff r^2 + r^2 \diff \Omega^2,
\end{equation*}
with  $\diff \Omega^2 = \diff \varphi^2 + \sin^2 \varphi \diff \theta^2$. The coordinate $r$ is only valid in the range $2M < r < + \infty$. We will consider the metric $g_\perp$, where the term containing $\diff \Omega^2$ is dropped. Denote by $h(r) = 1 - \frac{2M}{r}$ and introduce now the change of coordinates $v = t + r$ and $w = t - r$ to obtain:
\begin{equation*}
    g = \frac{1}{4} \Big( \frac{1}{h} - h \Big) \diff v^2 - \frac{1}{2} \Big( \frac{1}{h} + h \Big) \diff v \diff w + \frac{1}{4} \Big( \frac{1}{h} - h \Big) \diff v^2.
\end{equation*}
Observe that the condition $r \geqslant 0$ is equivalent to $v \geqslant w + 4M$.

Recall that the compactification of the configuration space by defining $\alpha = \arctan v$ y $\beta = \arctan w$. The range of both coordinates is $- \pi/2 \leqslant \alpha \leqslant \pi/2$ and $-\pi/2 \leqslant \beta \leqslant \pi/2$. Equality can be attained as we are modelling the compactified space as a manifold with corners. The condition $v \geqslant w + 4M$ is equivalent to $\tan \alpha \geqslant \tan \beta + 4M$. The region spanned by the coordinates $\alpha, \beta$ will be denoted by $N$. In these new coordinates, the metric becomes singular and reads
\begin{equation*}
    g_\perp = g = \frac{1}{4} \Big( \frac{1}{h} - h \Big) \sec^4 \alpha \diff \alpha^2 - \frac{1}{2} \Big( \frac{1}{h} + h \Big) \sec^2 \alpha \sec^2 \beta \diff \alpha \diff \beta + \frac{1}{4} \Big( \frac{1}{h} - h \Big) \sec^4 \beta \diff \beta^2.
\end{equation*}


The secant function blows up quadratically at $\alpha, \beta = \pi/2$ and the $E$-structure  is given by the set of vector fields which vanish quadratically at the boundary, described locally by values neighbouring $\alpha = \pi/2$. We will now compute the musical isomorphism induced by $g$ (showing, as a consequence, that it is non-degenerate) and the induced kinetic energy.

Let us consider a sufficiently small open chart $U \subset M$ centered around $\alpha_0 = \pi / 2$ and $\beta_0 \neq \pi/2$ with coordinates $\alpha, \beta$ satisfying $\beta \neq \pi/2$ at $U$. In this setting, the fields vanishing quadratically at the boundary (an early example of $E$-fields) admit as local generators $\alpha^2 \frac{\partial}{\partial \alpha}$ and $\frac{\partial}{\partial \beta}$ and the dual basis is written as $\frac{1}{\alpha^2} \diff \alpha, \diff \beta$. The metric can be written as:
\begin{equation*}
    g_\perp = \frac{1}{4} \Big( \frac{1}{h} - h \Big) \csc^4 \alpha \diff \alpha^2 - \frac{1}{2} \Big( \frac{1}{h} + h \Big) \csc^2 \alpha \sec^2 \beta \diff \alpha \diff \beta + \frac{1}{4} \Big( \frac{1}{h} - h \Big) \sec^4 \beta \diff \beta^2.
\end{equation*}
In the set of local generators of $E$-fields, the matrix of the metric and the inverse matrix is given by
\begin{align*}
    M_g &= \frac{1}{4} \begin{pmatrix}
        (h^{-1} - h) \alpha^4 \csc^4 \alpha & -(h^{-1} + h) \alpha^2 \csc^2 \alpha \sec^2 \beta \\ - (h^{-1} + h) \alpha^2 \csc^2 \alpha \sec^2 \beta & (h^{-1} - h) \sec^4 \beta
    \end{pmatrix} \\
    M_g^{-1} &= - \frac{4 \sin^4 \alpha \cos^4 \beta}{\alpha^4} \begin{pmatrix}
        (h^{-1} - h) \sec^4 \beta & (h^{-1} + h) \alpha^2 \csc^2 \alpha \sec^2 \beta \\ (h^{-1} + h) \alpha^2 \csc^2 \alpha \sec^2 \beta & (h^{-1} - h) \alpha^4 \csc^4 \alpha
    \end{pmatrix}.
\end{align*}
As $\lim_{\alpha \rightarrow 0} h = 1$ and $\alpha \csc \alpha$ can be extended to a smooth function around $\alpha = 0$, the metric is smooth and so is the inverse matrix. The induced kinetic, which is computed in natural coordinates $p_\alpha$ and $p_\beta$ in the $E$-cotangent bundle associated to the sections $\frac{1}{\alpha^2} \diff \alpha, \diff \beta$ as
\begin{align*}
    (g^\sharp)^*K(\alpha, \beta, p_\alpha, p_{\beta}) &= - 4 \frac{\sin^4 \alpha \cos^4 \beta}{\alpha^4} \Big( (h^{-1} - h) \sec^4 \beta p_\alpha^2 + 2 (h^{-1} + h) \alpha^2 \csc^2 \alpha \sec^2 \beta p_\alpha p_\beta \\
    & \qquad\qquad\qquad\qquad\qquad\qquad\qquad\qquad\qquad\qquad\qquad + (h^{-1} - h) \alpha^4 \csc^4 \alpha p_\beta^2 \Big),
\end{align*}
is also smooth.


After these computations, we can compute Wong's equations for an electromagnetic field in these coordinates. As $N$ is simply connected, every fibre bundle is trivial and, therefore, the principal $\mathrm{U}(1)$-bundle which describes the theory is $N \times \mathrm{U}(1)$. As $\mathrm{U}(1)$ is one-dimensional, the connection is completely specified by a pair of functions $A_1(\alpha, \beta)$ and $A_2(\alpha, \beta)$. Moreover, the structure constants $c_{ij}^k$ vanish. The minimal coupling of this system with $A_1(\alpha, \beta), A_2(\alpha, \beta)$ gives Wong's equations as in example \ref{ex:BMinCoupGen}.
\end{example}



In the physics terminology, the sections of the associated bundle $P \times_G \mathfrak{g}^*$ are referred to as \emph{matter fields}. Understanding a section of a bundle as a kind of topological generalization of a function, the choice of a matter field amounts to specifying a given element of $\mathfrak{g}^*$ for each space-time event $m \in M$. As elements of $\mathfrak{g}^*$ as the generalization of electric charges, a matter field corresponds to a density of charges on the space-time $M$. Our framework allows us to rigorously define the notion of such matter fields in manifolds with boundary, manifolds with corners, and manifolds with additional singularities, as well as studying their evolution under Wong's equations of motion.

In the same spirit, we can see the effect of Marsden-Weinstein's reduction by direct observation of Wong's equations. The charges of the model do not have conjugated momenta, contrasting with the classical position variables in the base manifold $M$. The quotient by gauge transformations identifies the momenta for internal charges and completely determines the evolution of the charge by coupling it with the canonical momenta.

\subsection*{Data availability statement}
Data sharing not applicable to this article as no datasets were generated or analysed during the current study.

\bibliographystyle{alpha}
\bibliography{bibliography}

\nocite{*}

\end{document}